\documentclass[a4paper,UKenglish]{article}
\usepackage[T1]{fontenc} 

\pdfoutput=1

\usepackage{hyperref}
\usepackage{graphicx}
\usepackage{epstopdf}

\usepackage{subcaption}

\usepackage{amsmath}
\usepackage{amsfonts}

\usepackage{calc}

\usepackage{caption}

\usepackage{amsthm}
\newtheorem{theorem}{Theorem}
\newtheorem{definition}{Definition}
\newtheorem{lemma}{Lemma}
\newtheorem{corollary}{Corollary}
\newtheorem{remark*}{Remark}
\newtheorem*{openprob*}{Open Problem}
\newtheorem*{conjecture*}{Conjecture}
\newtheorem*{example*}{Example}
\newtheorem*{note*}{Note}
\newtheorem*{prob*}{Problem}
\newtheorem*{acknowledgements*}{Acknowledgements}

\usepackage{tikz}
\usepackage{authblk}

\usepackage[ruled,vlined,nofillcomment,longend,linesnumbered]{algorithm2e}
\newlength\algowd

\usepackage{enumitem}

\allowdisplaybreaks
\title{The complexity of optimal design of temporally connected graphs.}

\author[1]{Eleni C. Akrida}
\author[1]{Leszek G\k{a}sieniec}
\author[2]{George B. Mertzios}
\author[1]{Paul G. Spirakis}
\affil[1]{Department of Computer Science, University of Liverpool, UK\\
  \texttt{\{Eleni.Akrida,L.A.Gasieniec,P.Spirakis\}@liverpool.ac.uk}}
\affil[2]{School of Engineering and Computing Sciences, Durham University, UK\\
  \texttt{George.Mertzios@durham.ac.uk}}

\begin{document}

\maketitle

\begin{abstract}
We study the design of small cost temporally connected graphs, under various constraints. We mainly consider undirected graphs of $n$ vertices, where each edge has an associated set of discrete availability instances (labels). A journey from vertex $u$ to vertex $v$ is a path from $u$ to $v$ where successive path edges have strictly increasing labels. A graph is temporally connected iff there is a $(u,v)$-journey for any pair of vertices $u,v,~u\not= v$. We first give a simple polynomial-time algorithm to check whether a given temporal graph is temporally connected. We then consider the case in which a designer of temporal graphs can \emph{freely choose} availability instances for all edges and aims for temporal connectivity with very small \emph{cost}; the cost is the total number of availability instances used. We achieve this via a simple polynomial-time procedure which derives designs of cost linear in $n$. We also show that the above procedure is (almost) optimal when the underlying graph is a tree, by proving a lower bound on the cost for any tree. However, there are pragmatic cases where one is not free to design a temporally connected graph anew, but is instead \emph{given} a temporal graph design with the claim that it is temporally connected, and wishes to make it more cost-efficient by removing labels without destroying temporal connectivity (redundant  labels).
Our main technical result is that computing the maximum number of redundant labels is APX-hard, i.e., there is no PTAS unless $P=NP$. On the positive side, we show that in dense graphs with random edge availabilities, there is asymptotically almost surely a very large number of  redundant labels. A temporal design may, however, be \emph{minimal}, i.e., no redundant labels exist. We show the existence of minimal temporal designs with at least $n \log{n}$ labels.
\end{abstract}

\section{Introduction and motivation}

A temporal network is, roughly speaking, a network that changes with time. A great variety of modern and traditional networks are not static and change over time. For example, social networks, wired or wireless networks may change dynamically, transport network connections may only operate at certain times, etc. Dynamic networks in general have been attracting attention over the past years~\cite{avin,xuan,casteigts,dutta,spirakisc}, exactly because they model real-life applications. In this work, following the model of~\cite{kempe, spirakis} and~\cite{akrida}, we consider \emph{discrete time} and restrict our attention to systems in which only the connections between the participating entities may change but the entities remain unchanged. So we consider networks, the links of which are available only at certain discrete time instances, e.g. days or hours. This is a natural assumption when the dynamicity of the system is inherently discrete, e.g., in synchronous mobile distributed systems that operate in discrete rounds. Moreover, it gives a purely combinatorial flavour to the resulting models and problems.

In several such dynamic settings, maintaining connections may come at a cost; consider the transport network example above or an unstable chemical or physical structure, where energy is required to keep a link available. We define the cost as the total number of discrete time instances at which the network links become available. We focus on design issues of temporal networks that are temporally connected; a temporal network is temporally connected if information can travel over time from any node to any other node following \emph{journeys}, i.e., paths whose successive edges have strictly increasing availability time instances. If one has absolute freedom to design a small cost temporally connected temporal network on an underlying static network, i.e, choose the edge availabilities, then a reasonable design would be to select a rooted spanning tree and choose appropriate availabilities to construct time-respecting paths from the leaves to the root and \emph{then} from the root back to the leaves. However, in more complicated scenarios one may not be free to \emph{choose} edge availabilities arbitrarily but instead \emph{specific} link availabilities might pre-exist for the network; then, one is able to design a temporally connected temporal network using only the pre-existing availabilities or a subset of them. Imagine a hostile network on a complete graph where availability of a link means a break in its security, e.g., when the guards change shifts, and only then are we able to pass a message through the link. So, if we wish to send information through the network, we may only use the times when the shifts change and it is reasonable to try and do so by using as few of these breaks as possible. In such scenarios, we may need to first verify that the pre-existing edge availabilities indeed define a temporally connected temporal network. Then, we may try to reduce the cost of the design by \emph{removing} unnecessary (redundant) edge availabilities if possible, without loosing temporal connectivity. Consider, again, the clique network of $n$ vertices with one time availability per edge; it is clearly temporally connected with cost $\Theta(n^2)$. However, it is not straightforward if all these edge availabilities are necessary for temporal connectivity. We resolve here the complexity of finding the maximum number of redundant labels in any given temporal graph.

\subsection{The model and definitions}
It is generally accepted to describe a network topology using a graph, the vertices and edges of which represent the communicating entities and the communication opportunities between them respectively. We consider graphs whose edge availabilities are described by sets of positive integers (labels), one set per edge.
\begin{definition}[Temporal Graph]
Let $G=(V,E)$ be a (di)graph. A temporal graph on $G$ is an ordered triple $G(L)=(V,E,L)$, where $L=\{L_{e} \subseteq \mathbb{N^*}:e\in E\}$ is an \emph{assignment} of labels to the edges (arcs) of $G$. $L$ is called a \emph{labelling} of $G$.
\end{definition}

\begin{definition}[Time edge]
Let $e=\{u,v\}$ (resp. $e=(u,v)$) be an edge (resp. arc) of the underlying (di)graph of a temporal graph and consider a label $l\in L_e$. The ordered triplet $(u,v,l)$ is called \emph{time edge}.
\end{definition}

Note that an undirected edge $e=\{u,v\}$ is associated with $2\cdot |L_e|$ time edges, namely both $(u,v,l)$ and $(v,u,l)$ for every $l\in L_e$.

The labels of an edge (arc) $e$ are the \emph{discrete time instances} at which $e$ is available. In many networks and in several applications, the availability of links comes at a cost. For example, in secure networks there is a cost (per discrete time instance) to keep a link secure. We abstract such considerations by the concept of the \emph{cost} of a temporal graph and wish to have temporal graphs of low cost.
\begin{definition}[Cost of a labelling]
Let $G(L)=(V,E,L)$ be a temporal (di)graph and $L$ be its labelling. The \emph{cost} of $L$ is defined as $c(L)= \sum_{e\in E} |L_e|$.
\end{definition}

A basic assumption that we follow here is that when a message or an entity passes through an available link at time $t$, then it can pass through a subsequent link only at some time $t'>t$ and only at a time at which that link is available. 
\begin{definition}[Journey]
A \emph{temporal path} or \emph{journey} $j$ from a vertex $u$ to a vertex $v$ {\emph ($(u, v)$-journey)} is a sequence of time edges $(u, u_1, l_1)$, $(u_1, u_2, l_2)$, $\ldots$ , $(u_{k-1}, v, l_k)$, such that $l_i < l_{i +1}$, for each $1 \leq i \leq k - 1$. We call the last time label, $l_k$, {\emph arrival time} of the journey.
\end{definition}
\begin{definition}[Foremost journey]
A $(u,v)$-journey $j$ in a temporal graph is called \emph{foremost journey} if its arrival time is the minimum arrival time of all $(u,v)$-journeys' arrival times, under the labels assigned to the underlying graph's edges. We call this arrival time the \emph{temporal distance}, $\delta(u,v)$, of $v$ from $u$.
\end{definition}

In this work, we focus on \emph{temporally connected} temporal graphs, i.e., temporal graphs that have the following property:
\begin{definition}[Property TC]
A temporal (di)graph $G(L) = (V,E,L)$ satisfies the property TC, or equivalently $L$ satisfies TC on $G$, if for any pair of vertices $u,v \in V,~u\not=v$, there is a $(u,v)$-journey \emph{and} a $(v,u)$-journey in $G(L)$. A temporal (di)graph that satisfies the property TC is called \emph{temporally connected}.
\end{definition}
\begin{example*}
An undirected complete graph, $K_n$, is temporally connected under any labelling $L$ with $L_e\not= \emptyset$ for every $e \in E(K_n)$. Indeed, there is a $(u,v)$-journey and a $(v,u)$-journey between any $u,v\in V(K_n),~u\not=v$, namely the time edge $(u,v,l)$ and the time edge $(v,u,l)$ respectively, for any $l \in L_{\{u,v\}}$.
\end{example*}

\begin{definition}[Minimal temporal graph]
A temporal graph $G(L)=(V,E,L)$ over a (strongly) connected (di)graph is \emph{minimal} if $G(L)$ has the property TC, and the removal of any label from any $L_e,~e\in E$, results in a $G(L')$ that \emph{does not} have the property TC.
\end{definition}
\begin{definition}[Removal profit]
Let $G(L)=(V,E,L)$ be a temporally connected temporal graph. The \emph{removal profit} $r(G,L)$ is the largest total number of labels that can be removed from $L $ without violating TC on $G$.
\end{definition}
Here, removal of a label $l$ from $L$ refers to the removal of $l$ only from a particular edge and not from all edges that are assigned label $l$, i.e., if $l \in L_{e_1} \cap L_{e_2}$ and we remove $l$ from both $L_{e_1}$ and $L_{e_2}$, it counts as two labels removed from $L$.

Notice that if many edges have the same label, we can encounter \emph{trivial cases} of minimal temporal graphs. For example, the complete graph where every edge appears at time, say $t=5$, is minimal but there are no journeys of length larger than $1$. To avoid cases where minimality is caused merely due to the assignment of the same label(s) to many (or all) edges, we will often consider a special sub-category of (single-labelled) temporal graphs:
\begin{definition}[SLSE temporal graphs]
A Single-label-single-edge (SLSE) temporal graph is a temporal graph, each edge of which has a single label and no two edges have the same label, i.e., each label is assigned to (at most) a single edge. A labelling that gives an SLSE temporal graph is also called \emph{SLSE labelling}.
\end{definition}

\subsection{Previous work and our contribution}\label{sec:full_paper_related}
In recent years, there is a growing interest in distributed computing systems that are inherently dynamic. For example, temporal dynamics of network flow problems were considered in a set of pioneering papers \cite{skutella1, skutella2,woeginger, tardos}. The model we consider here is very closely related to the single-labelled model of the seminal paper of~\cite{kempe} as well as the multi-labelled model of~\cite{spirakis}. In~\cite{kempe}, the authors consider the case of one \emph{real} label per edge and examine how basic graph properties change when we impose the temporal condition; here, we extend that model by considering multiple labels per edge but we restrict our focus to integer labels. In \cite{spirakis}, the model of~\cite{kempe} is also extended to many labels per edge and the authors mainly examine the number of labels needed for a temporal design of a network to guarantee several graph properties with certainty. The latter also defined the cost notion and, amongst other results, gave an algorithm to compute foremost journeys which can be used to decide property TC. However, the time complexity of that algorithm was pseudo-polynomial, as it was dominated by the cube of the maximum label used in the given labelling.

In fact, the problem of testing whether a dynamic graph is temporally connected has been studied before in various settings~\cite{xuan,whitbeck,barjon}. The authors of~\cite{xuan} propose an algorithm for computing foremost journeys in a model of evolving graphs, where nodes and edges are associated with lists of time intervals, representing their existence over time, and each edge has a traversal time. In a similar setting,~\cite{whitbeck} studies temporal reachability graphs, in which a $(u,v)$-edge is present at time $t$ if (in the corresponding time-varying graph) there is a $(u,v)$-journey leaving $u$ after $t$ and arriving at $v$ after at most some specified time-interval. In~\cite{barjon}, the authors investigate discrete-time evolving graphs, for which they compute the \emph{transitive closure of journeys}, i.e., a static directed graph whose edges represent potential journeys. The algorithm they propose depends on the maximum label used, the number of vertices, and the maximum number of edges that simultaneously exist.

Here, we show that if the designer of a temporal graph can select edge availabilities freely, then an asymptotically optimal linear-cost (in the size of the graph) design that satisfies TC can be easily obtained (cf. Section~\ref{sec:cost_opt_design}). We give a matching lower bound to indicate optimality, in the case where the underlying graph is a tree. However, there are pragmatic cases where one is not free to design a temporal graph anew; instead, one is \emph{given} a set of possible availabilities per edge with the claim that they satisfy TC and the constraint that they may only use them or a subset of them for their design. We also propose a simple algorithm to verify TC in low polynomial time (cf. Section~\ref{sec:foremost}). The \emph{given} design may also be minimal; we partially characterise minimal designs in Section~\ref{sec:minimal}. On the other hand, there may be some labels of the initial design that can be removed without violating TC (and also result in a lower cost). In this case, how many labels can we remove at best? Our main technical result is that this problem is APX-hard, i.e. it has no PTAS unless $P=NP$. On the positive side, we show that in the case of complete graphs and random graphs, if the labels are also assigned at random, there is aymptotically almost surely a very large number of labels that can be removed without violating TC. A preliminary version of this work appeared in the $13^{th}$ Workshop on Approximation and Online Algorithms, WAOA 2015~\cite{akrida-waoa}.

Stochastic aspects and/or survivability of network design were also considered in \cite{gupta,lap1, lap2}.

\subsubsection{Further related work}
Below, we provide a short survey of papers with studies on networks labelled by time units or segments, in addition to the ones mentioned above.
	
\noindent \textbf{Labelled Graphs.} Labelled graphs have been widely used both in Computer Science and in Mathematics, e.g.,~\cite{molloy}. 

\noindent \textbf{Continuous Availabilities (Intervals).} Some authors have assumed the availability of an edge for a whole time-interval [$t_1,t_2$] or multiple such time-intervals and not just for discrete moments as we assume here. Examples of such studies are~\cite{xuan, tardos, akrida-algo}.

\noindent \textbf{Dynamic Distributed Networks.} In recent years, there is a growing interest in distributed computing systems that are inherently dynamic~\cite{angluin, avin, casteigts, clementi, dutta, kuhn, spirakisb, spirakisc, o'dell, sch}.

\noindent \textbf{Distance labelling.} A distance labelling of a graph $G$ is an assignment of unique labels to vertices of $G$ so that the distance between any two vertices can be inferred from their labels alone~\cite{gavoille, katz}.

\noindent \textbf{Random labellings.} Random temporal networks have been considered before, e.g., in~\cite{chaintreau,clementi,akrida}. In~\cite{chaintreau}, the authors model opportunistic mobile networks as a type of random temporal networks, where each edge exists at each time-step with a fixed probability, and show a small diameter in general for that type of networks. In ~\cite{clementi}, the authors examine the speed of information dissemination in a type of dynamic graphs, where each edge exists at each time-step with some probability depending on whether it existed in the previous time-step. The \emph{Expected Temporal Diameter} of the model of (random) temporal graphs that we consider here was first examined in \cite{akrida}.

\section{A low polynomial time algorithm for deciding TC
}\label{sec:foremost}
In this section, we propose a simple polynomial-time algorithm which, given a temporal (di)graph $G(L)=(V,E,L)$ and a source vertex $s \in V$, computes a \emph{foremost} $(s,v)$-journey, for every $v \not= s$, if such a journey exists. Curiously enough, the previously known algorithm was pseudo-polynomial~\cite{spirakis}. Our algorithm significantly improves the running time. In fact, we conjecture it is optimal.

\begin{theorem}
Algorithm \ref{alg:foremost} satisfies the following, for every vertex $v \in V,~v\not=s$:
\begin{enumerate}[label=(\alph*)]
\item If $arrival\_time[v] < + \infty$, then there exists a foremost journey from $s$ to $v$, the arrival time of which is exactly $arrival\_time[v]$. This journey can be constructed by following the $parent[v]$ pointers in reverse order.
\item If $arrival\_time[v] = + \infty$, then no $(s,v)$-journey exists.
\item The time complexity of Algorithm \ref{alg:foremost} is dominated by the sorting time of the set of time edges (resp. time arcs).
\end{enumerate}
\end{theorem}
\begin{proof}[Proof sketch]
The algorithm actually considers each existing label in the sequence of time labels, from the smallest to the largest one. For each label considered, it computes the foremost journeys from $s$ which arrive at that time\footnote{One can prove this by induction.}. The algorithm examines each time edge (resp. time arc) exactly once.
\end{proof}

\begin{corollary}
The time complexity of Algorithm \ref{alg:foremost} is $O\big(c(L)\cdot \log{c(L)}\big)$.
\end{corollary}
\begin{proof}
The time complexity of the algorithm is dominated by the sorting time of $S(L)$. One can sort $S(L)$ by comparison-based sorting resulting in running time $O(|S(L)|\cdot \log{|S(L)|}) = O\big(c(L)\cdot \log{c(L)}\big)$.
\end{proof}

\begin{algorithm}[ht]
\caption{Foremost journey algorithm}
\label{alg:foremost}
\SetAlgoLined

\KwIn{A temporal (di)graph $G(L)=(V,E,L)$ of $n$ vertices, the set of all time edges (arcs) of which is denoted by $S(L)$; a designated source vertex $s \in V$}
\KwOut{A foremost $(s,v)$-journey from $s$ to all $v \in V \setminus \{s\}$, where such a journey exists; if no $(s,v)$-journey exists, then the algorithm reports it.}

Sort $S(L)$ in increasing order of labels \tcc*{Note that $|S(L)|=c(L)$}
Let $S'$ be the sorted array of time edges (resp. time arcs) according to time labels\;
$R:=\{s\}$  \tcc*{The set of vertices to which $s$ has a foremost journey}
$arrival\_time[s] := 0$\;
\For{all $v \in V \setminus \{s\}$}{
		$parent[v]:= \emptyset$\;
		$arrival\_time[v] := +\infty$\;
}
\For{all time edges (resp. time arcs) $(a,b,l)$ in the order given by $S'$}{
			\If{$a \in R$ \textbf{and} $b \not\in R$ \textbf{and} $arrival\_time[a]<l$ }{
				$parent[b]:= a$\;
				$arrival\_time[b] := l$\;
				$R := R \cup \{b\} $\;
			}		
}

\end{algorithm}

\begin{conjecture*}
We conjecture that any algorithm that computes journeys out of a vertex $s$ must sort the time edges (resp. time arcs) by their labels, i.e., we conjecture that Algorithm \ref{alg:foremost} is asymptotically optimal with respect to the running time.
\end{conjecture*}

Note that Algorithm \ref{alg:foremost} can even compute foremost $(s,v)$-journeys, if they exist, that \emph{start} from a given time $t_{start}>0$. Simply, one ignores the time edges (arcs) with labels smaller than the start time.

\section{Asymptotically cost-optimal design for TC in undirected graphs.}\label{sec:cost_opt_design}
In this section, we study temporal design issues on connected undirected graphs, so that the resulting temporal graphs are temporally connected. In this scenario, the designer has absolute freedom to choose the edge availabilities of the underlying graph.

\begin{lemma}\label{thm:star}
There is an infinite family of graphs $G_n$ of $n$ vertices, for which the cost of any labelling that satisfies TC is at least $2n-3$.
\end{lemma}
\begin{proof}
Consider the star graph of $n$ vertices, $n\geq 4$. Let $v_n$ be the root and $v_1, v_2, \ldots, v_{n-1}$ be the leaves. In any labelling on the star graph, which assigns only one label to two (or more) edges $(v_n, v_x),~(v_n,v_y),~x,y=1,2,\ldots,n-1,~x\not=y$, at least one of the vertices $v_x,v_y$ cannot reach the other via a journey. Therefore, any TC satisfying labelling on the star graph must assign at least $2$ labels to all edges of the graph, except possibly on one edge where it assigns a single label. The TC satisfying labelling which assigns labels $1,3$ to all edges except for one and label $2$ to the remaining edge has, therefore, minimum cost, namely $2n-3$ (cf.~Figure \ref{fig:star}). 

\begin{figure}[!htb]
\centering
\includegraphics[width=0.3\textwidth]{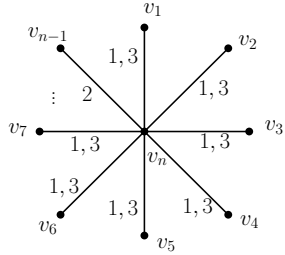}
\caption{Labelling a star graph in an optimal way}
\label{fig:star}
\end{figure}
\end{proof}


In fact, the result of Lemma \ref{thm:star} is optimal for any tree. Theorem~\ref{thm:thmdesign} shows a lower bound for trees and an asymptotically optimal\footnote{Any connected undirected graph needs at least $n-1$ labels on its edges to be temporally connected, and we show a TC satisfying labelling of $2n-2=\Theta(n)$ labels.} way of labelling any connected undirected graph to satisfy TC.

\begin{theorem}\label{thm:thmdesign}
\begin{enumerate}[label=(\alph*)]
\item For any tree $G=(V,E)$ of $n$ vertices and for any labelling $L$ that satisfies the property TC on $G$, the cost of $L$ is $c(L) \geq 2n-3$.
\item Given a connected undirected graph $G=(V,E)$ of $n$ vertices, we can design a labelling $L$ of cost $c(L) = 2(n-1)$ that satisfies the property TC on $G$. $L$ can be computed in polynomial time.
\end{enumerate}
\end{theorem}

\begin{proof}
\begin{enumerate}[label=(\alph*)]

\item\label{item:tree_lower_bound} We prove the statement by induction on the number of vertices of the tree.
\begin{description}
\item[\textbf{Base Case.}] \noindent It is easy to see that the statement holds for any tree of $n\leq 4$ vertices.

\item[\textbf{Induction Hypothesis.}] \noindent Assume that at least $2n-3$ labels are \emph{necessary} to satisfy TC on any tree of $n\leq k$ vertices, $k\in \mathbb{N}$.

\item[\textbf{Inductive Step.}] \noindent We will show that at least $2(k+1)-3=2k-1$ labels are necessary to satisfy TC on any tree of $k+1$ vertices.

Let $G=(V,E)$ be an arbitrary tree of $k+1$ vertices and let $L$ be an arbitrary labelling of $G$ that satisfies TC on $G$. Consider a leaf, $u \in V$, of $G$ and its unique neighbour, $u' \in V$. Note that $L$ must assign at least one label to the edge $\{u,u'\}$ to ``enable'' a journey between them. Now, let $L'$ be the sub-labelling of $L$ on $G\setminus u$. First, we show that, for $L$ to satisfy TC on $G$, it must be that $L'$ satisfies TC on $G \setminus u$.

Assume, to the contrary, that $L'$ does not satisfy TC on $G \setminus u$. Then, there exist two vertices $x,x' \in V(G\setminus u)$ such that the only journey(s) from $x$ to $x'$ in $G(L)$ go through $u$; let $J$ be a $(x,x')$-journey in $G(L)$. It must be:
\begin{eqnarray*}
J &=& \big(  (x,v_1,l_0), \ldots, (v_z,u',l_z), (u',u,l_{small}), \\
  & & (u,u',l_{big}), (u',v_{z'},l_{z'}), \ldots, (v_{last},x',l_{last}) \big) ,
\end{eqnarray*}
for some $v_0, \ldots, v_{last}$ and $l_0<\ldots< l_z<l_{small}<l_{big}< l_{z'} < \ldots <l_{last}$ (cf.~Figure~\ref{fig:journey_through_u}).
\begin{figure}[!htb]
\centering
\includegraphics[width=0.6\textwidth]{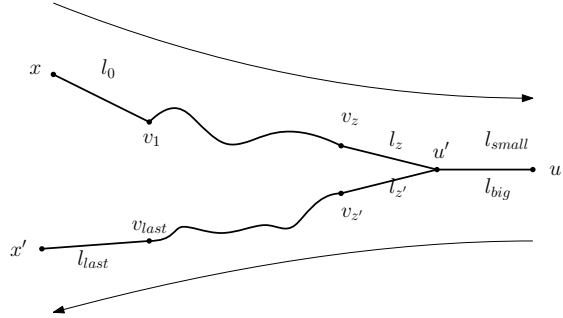}
\caption{A $(x,x')$-journey going through $u$.}
\label{fig:journey_through_u}
\end{figure}

But, then the sub-journey of $J$ which ``ignores'' the time-edges $(u',u,l_{small}),(u,u',l_{big})$ is still a $(x,x')$-journey in $G(L)$, which contradicts the fact that all $(x,x')$-journeys in $G(L)$ go through $u$. Therefore, $L'$ must satisfy TC on $G\setminus u$. Since $G\setminus u$ is a tree of $k$ vertices itself, it must be that $c(L')\geq 2k-3$ (by Induction Hypothesis).

If $c(L')\geq 2k-2$, then (since $L$ assigns at least one label to the edge $\{u,u'\}$), we have $c(L)\geq 2k-2 +1 = 2k-1$ and the Theorem holds.

It remains to check the case where $c(L')=2k-3$ and $L'$ satisfies TC on $G\setminus u$. $L'$ must assign at least one label to every edge of $G\setminus u$ to satisfy TC on it. Also, it must assign exactly one label to at least one edge $\{x,x'\} \in E(G\setminus u)$; if all edges of $G\setminus u$ had at least two labels under $L'$, then it would be $c(L') \geq 2(k-1)=2k-2$. Let $l_{unique}$ be the unique label of the edge $\{x,x'\}$. Also, without loss of generality, assume that $x$ is furthest from $u$ than $x'$ is, i.e., the unique path from $u$ to $x$ goes through $x'$. For $L$ to enable a $(u,x)$-journey in $G(L)$, it must assign to the edge $\{u,u'\}$ (at least) one label $l$ that is strictly smaller than $l_{unique}$. Also, to enable a $(x,u)$-journey, $L$ must assign to the edge $\{u,u'\}$ (at least) one label $l'$ that is strictly greater than $l_{unique}$ and, thus, different from label $l$ (cf~Figure~\ref{fig:journey_from_u_to_x}). So, $L$ assigns to $\{u,u'\}$ at least two labels, which makes the cost of $L$:
\[c(L) \geq c(L') +2 = 2k-3+2 = 2k-1\]

\begin{figure}[!htb]
\centering
\includegraphics[width=0.6\textwidth]{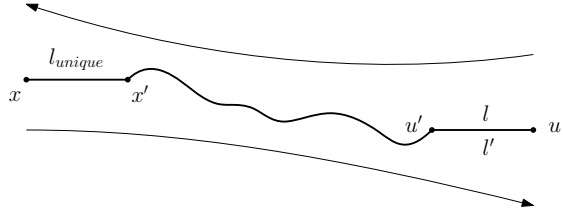}
\caption{$L$ must assign to $\{u,u'\}$ at least $2$ labels.}
\label{fig:journey_from_u_to_x}
\end{figure}

Therefore, in any case, for $L$ to satisfy TC on $G$, it needs to have cost $c(L) \geq 2k -1$.
\end{description}

%
%
%
%

\item\label{item:labelling} Consider a fixed, but arbitrary, spanning tree $T$ of $G$ and let $w$ be the root of $T$. Let $r$ be the length of the longest path from $w$ to any leaf of $T$, i.e., $r$ is the radius of $T$. We assign labels to the edges of $T$ as follows:
\begin{description}
\item[\textbf{Going upwards.}] Any edge incident to a leaf of $T$ gets label $1$. Any edge $e=\{u,v\}$, with $d(w,v) = d(w,u)+1$, where the subtree $T'$ rooted at $v$ has been labelled going upwards to the root, gets a label $l_e=max\{ \text{all labels in } T' \}+1$~(cf.~Figure~\ref{fig:subtree}).
\begin{figure}[!htb]
\centering
\includegraphics[width=0.3\textwidth]{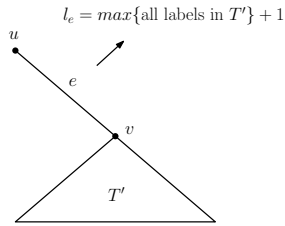}
\caption{Labelling ``going upwards'' to the root}
\label{fig:subtree}
\end{figure}
\item[\textbf{Going downwards.}] Any edge incident to the root gets a label $r+1$. Any edge $e$ in a path from the root to a leaf, the \emph{parent edge}\footnote{The edge before it in the sequence of edges from the root to the respective leaf.} of which has been labelled, going downwards, with label $l'$, gets a label $l_e = l'+1$.
\end{description}
We can easily implement the above process by topologically ordering the vertices of $T$ in levels using \emph{Breadth First Search} and implement the ``going upwards'' and ``going downwards'' procedures accordingly. The above method results in a labelling where:
\begin{enumerate}[label=\arabic*., ref=\arabic*]
\item each edge of $T$ has $2$ labels,
\item each edge of $E\setminus T$ has no label and
\item\label{item:c} for each ordered pair of vertices $u,v\in V,~u\not=v$, there is a $(u,v)$-journey.
\end{enumerate}
To show \ref{item:c}, just notice that one can go from any vertex $u\in V$ to any other vertex $v \in V$ by going up in $T$ from $u$ to $w$ and then going down in $T$ from $w$ to $v$ via strictly increasing labels, by construction.
\end{enumerate}
\end{proof}

\begin{example*} Figure \ref{fig:example} shows an example of the procedure described above. Notice the existence of journeys from any vertex to every other vertex in the resulting temporal graph.
\begin{figure}[!htb]
\centering
\includegraphics[width=0.7\textwidth]{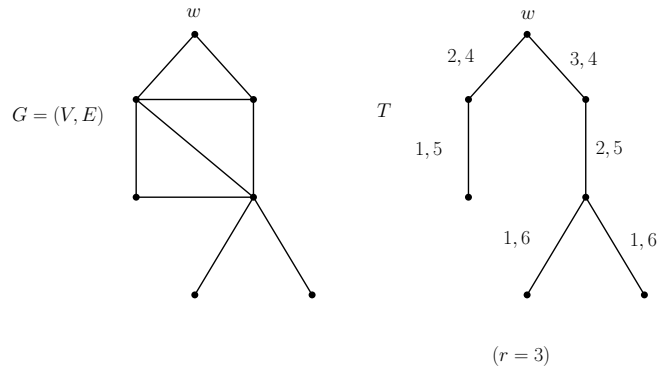}
\caption{Labelling a connected undirected graph to satisfy TC}
\label{fig:example}
\end{figure}
\end{example*}

\begin{conjecture*}
We conjecture that for any connected undirected graph $G$ of $n$ vertices and for any labelling $L$ that satisfies the property TC on $G$, the cost of $L$ is $c(L)\geq 2n-4$.
\end{conjecture*}

\section{Minimal Temporal Designs}\label{sec:minimal}
Suppose now that a temporal graph on a (strongly) connected (di)graph $G=(V,E)$ is \emph{given} to a designer with the claim that it satisfies TC. In this scenario, the designer is allowed to only use the given set of edge availabilities, or a subset of them. If the given design is not minimal, they may wish to remove as many labels as possible, thus reducing the cost. Minimality of a design can be verified by running Algorithm \ref{alg:foremost} (cf.~Section \ref{sec:foremost}) for every $s\in V$.


\subsection{A partial characterisation of minimal temporal graphs}\label{sec:full_paper_hyper}
As mentioned earlier, if many edges have the same label, we can encounter \emph{trivial cases} of minimal temporal graphs. To avoid such cases, we focus our attention here to the class of SLSE temporal graphs, in which every edge only becomes available at one moment in time and no two different edges become available at the same time. Are there minimal SLSE temporal graphs with non linear (in the size of the graph) cost? For example, any complete SLSE temporal graph satisfies TC. Are all these $\Theta(n^2)$ labels needed for TC, i.e., are there minimal temporal complete graphs? As we prove in Theorem \ref{thm:clique_code}, the answer is negative. However, we give below a minimal temporal graph on $n$ vertices with non-linear in $n$ cost, namely with $O(n\log{n})$ labels.

\subsubsection{A minimal temporal design of \texorpdfstring{$n \log{n}$}{nlogn} cost}\label{sec:nlogn_minimal}

\begin{definition}[Hypercube graph]
The $k$-hypercube graph, commonly denoted $Q_k$, is a $k$-regular graph of $2^k$ vertices and $2^{k-1} \cdot k$ edges. The $1$-hypercube is the graph of two vertices and one edge. Recursively, the $n$-hypercube is produced by taking two isomorphic copies of the $(n-1)$-hypercube and adding edges between the corresponding vertices.
\end{definition}

\begin{definition}[Flat]
In geometry, a \emph{flat} is a subset of the $n$-dimensional space that is congruent to a Euclidean space of lower dimension, e.g., the flats in the two-dimensional space are points and lines. In the $n$-dimensional space, there are flats of every dimension from $0$, i.e., points, to $n-1$, i.e., hyperplanes.
\end{definition}
\begin{theorem}
There exists an infinite class of minimal temporal graphs on $n$ vertices with $\Theta(n\cdot \log{n})$ edges and $\Theta(n\cdot \log{n})$ labels, such that different edges have different labels.
\end{theorem} 
\begin{proof}
We present a minimal temporal graph on the hypercube graph of $n$ vertices. Consider Protocol \ref{label_protocol'} for labelling the edges of $G=Q_k=(V,E)$. The temporal graph, $G(L)$, that this labelling procedure produces on the hypercube is minimal. Indeed, first we will prove that the temporal graph produced by Protocol \ref{label_protocol'} satisfies TC on $G=Q_k$.

Consider vertices $u,v\in V$ and the steps described in Protocol \ref{journeyQn} to reach $v$, starting from $u$, via temporal edges. The procedure described in Protocol \ref{journeyQn} gives a journey from $u$ to $v$, which is also \emph{unique}. It suffices to consider the $k$-bit binary representation of the vertices of $G$. Notice that if the hamming distance of the labels of two vertices $u,v\in V(G)$ is exactly $m$, then to reach $v$ from $u$ via a temporal path in the temporal graph on $G$, we need to move through vertices by consecutively swapping the bits in which $u$ and $v$ differ in the order of dimensions. This way, we maintain the strictly increasing order of the time labels we use and, swap by swap, we approach the destination. Note also that swapping only the bits in which $u$ and $v$ differ is the only way to not violate the increasing order of time labels we use: without loss of generality, suppose that the $j^{th}$ bit of $u$ is $1$ and so is $j^{th}$ bit of $v$. If, starting from $u$, we swap the $j^{th}$ bit to $0$, i.e., we use an edge, $e$, on the $j^{th}$ dimension, then at a future step, we again need to swap the $j^{th}$ bit back to $1$ (otherwise, we never reach $v$). However, the two swaps cannot be consecutive, because then we would use edge $e$ twice and we violate the increasing order of labels. So, we would need to move to a higher dimension after the first of the two swaps; but, then, we have used labels that are larger than all the labels of the $j^{th}$ dimension, so using any edge of the $j^{th}$ dimension would also violate the increasing order of labels.
\begin{algorithm}[ht]
\SetAlgorithmName{Protocol}{protocol}{List of protocols}
\caption{Labelling the hypercube graph, $G=Q_k$}
\label{label_protocol'}
\SetAlgoLined

	Consider the $k$ dimensions of the hypercube $G=Q_k$, $x_1,x_2,\ldots,x_k$\;
	\For{$i= 1 \ldots k$}{
			Let $X_i:=\{e_{i1},e_{i2},\ldots,e_{i2^{k-1}}\}$ be the list of edges in dimension $x_i$, in an arbitrary order\;
			Let $L_i$ be the (sorted from smallest to largest) list of labels $L_i:=\{ (i-1) \cdot 2^{k-1} +1, (i-1) \cdot 2^{k-1} +2, \ldots, i \cdot 2^{k-1} \}$ \;
	}	
	\For{$i= 1 \ldots k$}{
		\For{$j= 1 \ldots 2^{k-1}$}{
				Assign the (current) first label of $L_i$ to the (current) first edge of $X_i$ \;
				Remove the (current) first label of $L_i$ from the list\;
				Remove the (current) first edge of $X_i$ from the list\;
		}
	}
	\Return{the produced temporal graph, $G(L)$;}
\end{algorithm}

\begin{algorithm}[ht]
\SetAlgorithmName{Protocol}{protocol}{List of protocols}
\caption{A temporal path from $u$ to $v$ in the temporal graph on $G=Q_k$}
\label{journeyQn}
\SetAlgoLined

	\KwIn{The considered temporal graph on the hypercube $G=Q_k$, vertices $u,v \in V(G)$}
	\KwOut{Array $x$ of vertices, which the $(u,v)$-journey passes through}
	$x[0]:=u$\;
	Find the flat of the smallest dimension, $m$, which both $u$ and $v$ lie on\;
	Consider the increasing order of the $m$ dimensions in that flat: $d[1], d[2],\ldots, d[m]$\;
	\For{$i=1 \ldots m$}{
			Use the incident edge of $x[i-1]$ that lies on dimension $d[i]$ and let $x[i]$ be the other endpoint of that edge\;
	}
\end{algorithm}

Since our labelling gives a \emph{unique} $(u,v)$-journey, for every $u,v\in V$, and since all labels assigned to the edges of $E$ are used in the union of all those journeys, the deletion of any single label will violate TC. Therefore, $G(L)$ is minimal. Finally, note that the temporal graph $G(L)$ on the hypercube graph $G=Q_k$ has $n=2^k$ vertices, $\frac{1}{2}n\cdot \log{n}$ edges and $\frac{1}{2}n\cdot \log{n}$ labels.
\end{proof}

\subsubsection{A minimal temporal design of linear in \texorpdfstring{$n$}{n} cost}
In the previous section, we showed that there are graphs of non-linear cost (in the number of vertices) that are minimal. Here, we show that there are classes of minimal graphs whose cost is linear in the number of their vertices.

Indeed, as seen in Lemma~\ref{thm:star} (Section~\ref{sec:cost_opt_design}), the star graph of $n$ vertices needs at least $\Theta(n)$ labels to satisfy TC and, in fact, we present there a TC satisfying labelling of $\Theta(n)$ labels (cf.~Figure~\ref{fig:star}). Theorem~\ref{thm:thmdesign}\ref{item:labelling} (Section~\ref{sec:cost_opt_design}) also gives a class of minimal temporal graphs of linear cost in the number of vertices. Therefore, we have the following Corollary:

\begin{corollary}
There exists an infinite class of minimal temporal graphs on $n$ vertices with $\Theta(n)$ edges and $\Theta(n)$ labels.
\end{corollary}

\subsubsection{SLSE Cliques of at least 4 vertices are not minimal}\label{sec:clique_minimal}

The complete graph on $n$ vertices, $K_n$, with an SLSE labelling $L$, i.e., a labelling that assigns a single label per edge, different labels to different edges, is an interesting case, since $K_n(L)$ always satisfies TC. However, it is not minimal as the theorem below shows.

\begin{theorem}\label{thm:clique_code}
Let $n\in \mathbb{N},~n\geq 4$ and denote by $K_n$ the complete graph on $n$ vertices. There exists \emph{no minimal} SLSE temporal graph on $K_n(L)$. In fact, we can remove (at least) $\lfloor \frac{n}{4} \rfloor$ labels from any SLSE labelling on $K_n(L)$ without violating TC.
\end{theorem}
\begin{proof}
The proof is divided in two parts, as follows:
\begin{enumerate}[label=(\alph*)]

\item\label{item:code_a} We first show that any SLSE labelling on the complete graph on $4$ vertices produces a temporal graph that is not minimal, i.e., the theorem holds for $K_4$. Consider the six different labels $a,b,c,d,x,y$ assigned by an SLSE labelling to the edges of $K_4$ as shown in Figure~\ref{fig:K_4}.

\begin{figure}[tbh!]
\centering
\includegraphics[scale=0.7]{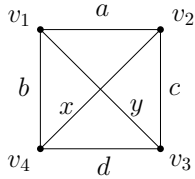}
\caption{Any SLSE labelling on $K_4$ is not minimal.}
\label{fig:K_4}
\end{figure}

Up to their renaming, there are three possible cases for the labels $a,b,c,d$. Counting all the cases of \emph{alternation}, \emph{cycle}, and \emph{entanglement} (see below) would give us all possible $4!=24$ cases.

\begin{enumerate}[label=\arabic*.]
\item (Alternation) $a<b>d<c>a$.\\It is easy to see that in this case, both diagonals can be removed: $v_1$ can reach $v_3$ using labels $a$ and then $c$; $v_3$ can reach $v_1$ using labels $d$ and then $b$; $v_2$ can reach $v_4$ using labels $a$ and then $b$; $v_4$ can reach $v_2$ using labels $d$ and then $c$.
\item (Cycle) $a<b<d<c$.\\Here, diagonal $x$ can be removed: $v_2$ can reach $v_4$ using labels $a$ and then $b$; $v_4$ can reach $v_2$ using labels $d$ and then $c$.
\item (Entanglement) $a<b<c<d$.\\This is a more complex case, for which we distinguish the following five sub-cases:
	\begin{enumerate}[label=\roman*)]
	\item $x<b$ and $y<c$.\\We can remove label $a$: $v_1$ can reach $v_2$ using labels $y$ and then $c$; $v_2$ can reach $v_1$ using labels $x$ and then $b$.
	\item $x<b$ and $y>c$.\\We can remove label $b$: $v_1$ can reach $v_4$ using labels $a$, then $c$ and then $d$; $v_4$ can reach $v_1$ using labels $x$, then $c$ and then $y$ (notice that $x<b<c<y$).
	\item $x>b$ and $y>c$.\\We can remove label $a$: $v_1$ can reach $v_2$ using labels $b$ and then $x$; $v_2$ can reach $v_1$ using labels $c$ and then $y$.
	\item $x>b$ and $b<y<c$.\\We can remove label $x$: $v_2$ can reach $v_4$ using labels $a$ and then $b$; $v_4$ can reach $v_2$ using labels $b$, then $y$ and then $c$.
	\item $x>b$ and $y<b$.\\We can remove label $c$: $v_2$ can reach $v_3$ using labels $a$, then $b$ and then $d$; $v_3$ can reach $v_2$ using labels $y$, then $b$ and then $x$.
	\end{enumerate}
	Notice that the coverage of the above five cases is complete (cf.~Figure~\ref{fig:coverage}).
	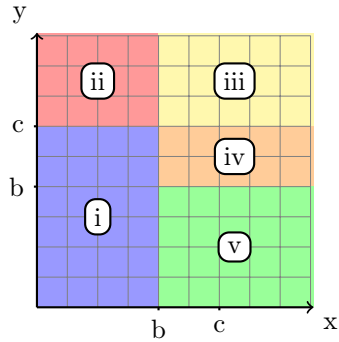
\begin{figure}[tbh!]
	\centering

	\begin{tikzpicture}[scale=0.4]
	
	\fill[blue!40!white] (0,0) rectangle (4,6);	
	\fill[red!40!white] (0,6) rectangle (4,9.1);
	\fill[yellow!40!white] (4,6) rectangle (9.1,9.1);
	\fill[orange!40!white] (4,4) rectangle (9.1,6);
	\fill[green!40!white] (4,0) rectangle (9.1,4);

	
	\draw[step=1cm,gray,very thin] (0,0) grid (9,9);
	\draw[thick,->] (0,0) -- (9.1,0) node[anchor=north west] {x};
	\draw[thick,->] (0,0) -- (0,9.1) node[anchor=south east] {y};
	
	\node[rectangle, rounded corners, draw, thick, fill=white] (i) at (2,3) {i};
	\node[rectangle, rounded corners, draw, thick, fill=white] (iii) at (2,7.5) {ii};
	\node[rectangle, rounded corners, draw, thick, fill=white] (ii) at (6.5,7.5) {iii};
	\node[rectangle, rounded corners, draw, thick, fill=white] (v) at (6.5,5) {iv};
	\node[rectangle, rounded corners, draw, thick, fill=white] (iv) at (6.5,2) {v};
	
	\draw[thick] (4,0) -- (4,-0.1) node[anchor=north] {b};
	\draw[thick] (6,0) -- (6,-0.1) node[anchor=north] {c};
	\draw[thick] (0,4) -- (-0.1,4) node[anchor=east] {b};
	\draw[thick] (0,6) -- (-0.1,6) node[anchor=east] {c};
	
	\end{tikzpicture}

	\caption{The six sub-cases cover all possible scenarios of ``entanglement''.}
	\label{fig:coverage}
	\end{figure}
\end{enumerate}

\item Now, consider the complete graph on $n\geq 4$ vertices, $K_n=(V,E)$. Partition $V$ arbitrarily into $\lceil \frac{n}{4} \rceil$ subsets $V_1, V_2, \ldots, V_{\lceil \frac{n}{4} \rceil}$, such that $|V_i|=4, \forall i=1,2, \ldots, \lceil \frac{n}{4} \rceil -1 $ and $|V_{\lceil \frac{n}{4} \rceil}| \leq 4$. In each $4$-clique defined by $V_i,~i=1,2,\ldots, \lfloor \frac{n}{4} \rfloor$, we can remove a ``redundant'' label, as shown in~\ref{item:code_a}. The resulting temporal graph on $K_n$ still preserves TC since for every ordered pair of vertices $u,v \in V$:
\begin{itemize}
\item if $u,v$ are in the same $V_i$, $i=1,2,\ldots, \lfloor \frac{n}{4} \rfloor$, then there is a $(u,v)$-journey that uses time edges within the 4-clique on $V_i$, as proven in \eqref{item:code_a}.
\item if $u\in V_i$ and $v\in V_j,~i\not = j$, then there is a $(u,v)$-journey that uses the (direct) time edge on $\{u,v\}$.
\end{itemize}
\end{enumerate}
\end{proof}

\subsection{Computing the removal profit is APX-hard}

Note that it is straightforward to check in polynomial time whether a given $L$ satisfies TC on a given (di)graph $G$, by just checking for every possible
(ordered) pair $(u,v)$ of vertices in $G$ whether there is a $(u,v)$-journey in $G(L)$. Recall that the removal profit is the largest number of labels that can be removed from a temporally connected graph without destroying TC. We now show that it is hard to approximate the value of the removal profit arbitrarily well for an arbitrary
graph, i.e.,~there exists no PTAS\footnote{PTAS stands for Polynomial-Time Approximation Scheme.} for this problem, unless P=NP. It is worth noting here that, in our hardness proof below, we consider \emph{undirected} graphs; the fact that all $(u,v)$-journeys, $u\not=v$ exist in any given (unlabelled) connected undirected graph makes the reduction and the analysis much more involved.

We prove our hardness result by providing an approximation preserving
polynomial reduction from a variant of the maximum satisfiability problem,
namely from the \emph{monotone Max-XOR($3$)} problem. Consider a monotone
XOR-boolean formula $\phi $ with variables $x_{1},x_{2},\ldots ,x_{n}$,
i.e.,~a boolean formula that is the conjunction of XOR-clauses of the form $%
(x_{i}\oplus x_{j})$, where no variable is negated. The clause $\alpha
=(x_{i}\oplus x_{j})$ is XOR-satisfied by a truth assignment $\tau $ if and
only if $x_{i}\neq x_{j}$ in $\tau $. The number of clauses of $\phi $ that
are XOR-satisfied in $\tau $ is denoted by $|\tau (\phi )|$. If every
variable $x_{i}$ appears in exactly $r$ XOR-clauses in $\phi $, then $\phi $
is called a \emph{monotone XOR(}$r$\emph{)} formula. The \emph{monotone
Max-XOR(}$r$\emph{)} problem is, given a monotone XOR($r$) formula $\phi $,
to compute a truth assignment $\tau $ of the variables $x_{1},x_{2},\ldots
,x_{n}$ that XOR-satisfies the largest possible number of clauses, i.e.,~an
assignment $\tau $ such that $|\tau (\phi )|$ is maximized. The monotone
Max-XOR($3$) problem essentially encodes the \emph{Max-Cut} problem on $3$%
-regular (i.e.,~cubic) graphs, which is known to be APX-hard \cite{Alimonti}.

\begin{lemma}\hspace{-0,01cm}\protect\cite{Alimonti}
\label{Max_XOR-3-hard-lem}The monotone Max-XOR($3$) problem is APX-hard.
\end{lemma}

Now we provide our reduction from the monotone Max-XOR($3$) problem to the
problem of computing $r(G,L )$. Let $\phi $ be an arbitrary monotone XOR($3$) formula 
with $n$ variables $x_{1},x_{2},\ldots,x_{n}$ and $m$ clauses. 
Since every variable $x_{i}$ appears in $\phi $ in
exactly $3$ clauses, it follows that $m=\frac{3}{2}n$. We will construct
from $\phi $ a graph $G_{\phi }=(V_{\phi },E_{\phi })$ and a
labelling $L _{\phi }$ of $G_{\phi }$.

First we construct for every variable $x_{i}$, where $1\leq i \leq n$, 
the gadget-graph~$G_{\phi,i}$ together with a labelling $L _{\phi ,i}$ of its edges, 
as illustrated in~Figure~\ref{removal-variable-gadget-fig}. 
In this figure, the labels of every edge in $L _{\phi ,i}$ are drawn next to the edge. 
We call the induced subgraph of $G_{\phi ,i}$ on the $4$ vertices $%
\{s^{x_{i}},u_{0}^{x_{i}},w_{0}^{x_{i}},v_{0}^{x_{i}}\}$ the \emph{base} of $%
G_{\phi ,i}$. Moreover, for every $p\in \{1,2,3\}$, we call the induced
subgraph of $G_{\phi ,i}$ on the $4$ vertices $%
\{t_{p}^{x_{i}},u_{p}^{x_{i}},w_{p}^{x_{i}},v_{p}^{x_{i}}\}$ the $p$\emph{th
branch} of $G_{\phi ,i}$. Finally, we call the edges $\{  u_{0}^{x_{i}} , w_{0}^{x_{i}}  \}$ and $ \{   w_{0}^{x_{i}} ,  v_{0}^{x_{i}} \}$ the \emph{%
transition edges} of the base of $G_{\phi ,i}$ and, for every $p\in
\{1,2,3\} $, we call the edges $  \{  u_{p}^{x_{i}} , w_{p}^{x_{i}}  \}$ and $%
\{  w_{p}^{x_{i}}  ,  v_{p}^{x_{i}}  \}$ the \emph{transition edges} of the $p$th branch of $G_{\phi ,i}$. 
For every $p\in \{1,2,3\}$ we associate the $p$th appearance of the variable $x_{i}$  with the $p$th
branch of~$G_{\phi ,i}$.

We continue the construction of $G_{\phi ,i}$ and $L _{\phi ,i}$ as
follows. First, we add an edge between any possible pair of vertices $%
w_{p}^{x_{i}},w_{q}^{x_{j}}$, where $p,q\in \{0,1,2,3\}$ and $i,j\in
\{1,2,\ldots ,n\}$, and we assign to this new edge $%
e=  \{  w_{p}^{x_{i}}  ,  w_{q}^{x_{j}}  \} $ the unique label~$L _{\phi }(e)=\{7\}$.
Note here that we add this edge $ \{ w_{p}^{x_{i}}  , w_{q}^{x_{j}}  \}$ also in the
case where $i=j$ (and $p\neq q$).

\begin{figure}[tbh]
\centering\includegraphics[scale=0.6]{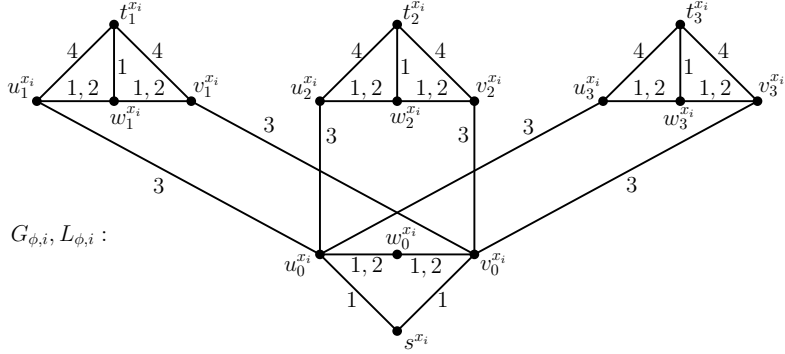}
\caption{The gadget $G_{\protect\phi ,i}$ for the variable $x_{i}$.}
\label{removal-variable-gadget-fig}
\end{figure}

Intuitively, the base of $G_{\phi,i}$ (cf.~Figure~\ref{removal-variable-gadget-fig}) 
corresponds to the variable $x_{i}$ and, 
for every $p\in\{1,2,3\}$, the $p$th branch of~$G_{\phi,i}$, together with the two edges $\{u_{0}^{x_{i}},u_{p}^{x_{i}}\}$ and $\{v_{0}^{x_{i}},v_{p}^{x_{i}}\}$, 
correspond to the clause of $\phi$ in which $x_{i}$ appears for the $p$th time in $\phi$.

Consider now a clause ${\alpha =(x_{i}\oplus x_{j})}$ of $\phi $. Assume that
the variable $x_{i}$ (resp.~$x_{j}$) of $\alpha $ corresponds to
the $p$th (resp.~to the $q$th) appearance of $x_{i}$ (resp.~of~$x_{j}$) in~$%
\phi $. Then we identify the vertices ${u_{p}^{x_{i}},v_{p}^{x_{i}},w_{p}^{x_{i}},t_{p}^{x_{i}}}$ 
of the $p$th branch of~$G_{\phi ,i}$ with the vertices $%
v_{q}^{x_{i}},u_{q}^{x_{i}},w_{q}^{x_{i}},t_{q}^{x_{i}}$ of the $q$th branch
of $G_{\phi ,j}$, respectively (cf.~Figure~\ref{removal-clause-gadget-fig}). 
Now we add an edge between any
possible pair of vertices $t_{p}^{x_{i}},t_{q}^{x_{j}}$, $%
i,j\in \{1,2,\ldots ,n\}$, and $p,q\in \{1,2,3\}$.
We assign to this new edge $e=  \{  t_{p}^{x_{i}}  ,  t_{q}^{x_{j}}  \}$ the unique label~$L _{\phi}(e)=\{7\}$.

Furthermore, for every $i\in \{1,2,\ldots ,n\}$ and every $p\in
\{1,2,3\}$ we define for simplicity of notation the temporal paths $%
P_{i,p}=(s^{x_{i}},u_{0}^{x_{i}},u_{p}^{x_{i}},t_{p}^{x_{i}})$ and $%
Q_{i,p}=(s^{x_{i}},v_{0}^{x_{i}},v_{p}^{x_{i}},t_{p}^{x_{i}})$.

The intuition behind the composition of the gadget-graphs~$G_{\phi,i}$ 
(cf.~Figure~\ref{removal-clause-gadget-fig}) is the following. 
If variable $x_{i}$ is false in a truth assignment $\tau$ of $\phi$, 
then all edges of the paths $P_{i,1},P_{i,2},P_{i,3}$ keep their labels as in $L_{\phi}$. 
Otherwise, if $x_{i}$ is true in $\tau$, 
then all edges of the paths $Q_{i,1},Q_{i,2},Q_{i,3}$ keep their labels as in $L_{\phi}$. 
Furthermore, depending on the value of $x_{i}$ in the assignment $\tau$, 
each of the transition edges $\{u_{p}^{x_{i}},w_{p}^{x_{i}}\}$ and $\{w_{p}^{x_{i}},v_{p}^{x_{i}}\}$, 
where $p\in\{1,2,3\}$, keeps exactly one of its two labels from $L_{\phi}$. 
Consider now a clause ${\alpha =(x_{i}\oplus x_{j})}$ of $\phi$ which corresponds to the 
$p$th branch of $G_{\phi,i}$ and to the $q$th branch of $G_{\phi,j}$. 
Then the only case where \emph{both} edges 
$\{t_{p}^{x_{i}},u_{p}^{x_{i}}\}$ and $\{t_{p}^{x_{i}},v_{p}^{x_{i}}\}$ keep their labels from $L_{\phi}$,
is when the two variables $x_{i},x_{j}$ have \emph{equal} truth value in the 
corresponding truth assignment $\tau$ of $\phi$; 
that is, when the clause ${\alpha =(x_{i}\oplus x_{j})}$ is \emph{not} XOR-satisfied by $\tau$. 
Therefore, intuitively, by a careful counting of the labels it turns out that, if more clauses can be satisfied by a truth assignment $\tau$, then a TC preserving sub-labelling $L$ of $L_\phi$ can be constructed which avoids more labels from $L_{\phi}$, and vice versa (cf.~Theorem~\ref{cost-removing-labels-upper-lower-bound-thm}).

To finalize the construction of the graph $G_{\phi}$, we add a new vertex $t_{0}$ to ensure 
the existence of a temporal path between each pair of vertices of $G_{\phi}$, as follows. 
This new vertex $t_{0}$ is adjacent to vertex $w_{0}^{x_{n}}$ 
and to all vertices in the 
set $\{s^{x_{i}},t_{1}^{x_{i}},t_{2}^{x_{i}},t_{3}^{x_{i}},u_{p}^{x_{i}},v_{p}^{x_{i}}:1\leq i\leq n,\ 0\leq p\leq 3\} 
$. First we assign to the edge $\{  t_{0}  ,  w_{0}^{x_{n}}  \}$ the unique label~$%
L _{\phi }(  \{  t_{0}  ,  w_{0}^{x_{n}}  \} )=\{5\}$. Furthermore, for every vertex $%
t_{p}^{x_{i}}$, where $1\leq i\leq n$ and $1\leq p\leq 3$, we assign to the
edge $  \{  t_{0}  ,  t_{p}^{x_{i}}  \} $ the unique label~$L _{\phi
}(  \{  t_{0}  ,  t_{p}^{x_{i}}  \} )=\{5\}$. Finally, for each of the vertices $z\in
\{s^{x_{i}},u_{p}^{x_{i}},v_{p}^{x_{i}}:1\leq i\leq n,\ 0\leq p\leq 3\}$ we
assign to the edge $ \{ t_{0}  ,  z  \} $ the unique label~$L _{\phi }(  \{  t_{0}  ,  z  \}  )=\{6\}$. 
The addition of the vertex $t_{0}$ and the labels of the (dashed) edges
incident to $t_{0}$ are illustrated in~Figure \ref{removal-t0-vertex-gadget-fig}.
Denote the vertex sets $A=%
\{s^{x_{i}},u_{p}^{x_{i}},v_{p}^{x_{i}}:1\leq i\leq n,\ 0\leq p\leq 3\}$, $%
B=\{w_{p}^{x_{i}}:1\leq i\leq n,\ 0\leq p\leq 3\}$, and $C=\{t_{p}^{x_{i}}:1%
\leq i\leq n,\ 1\leq p\leq 3\}$. 
Note that $V_{\phi }=A\cup B\cup C\cup \{t_{0}\}$. 
This completes the construction of the graph $G_{\phi }$ and its labelling $L _{\phi }$.

\begin{figure}[tbh]
\centering
\begin{subfigure}[t]{.8\linewidth}
\centering \includegraphics[scale=0.6]{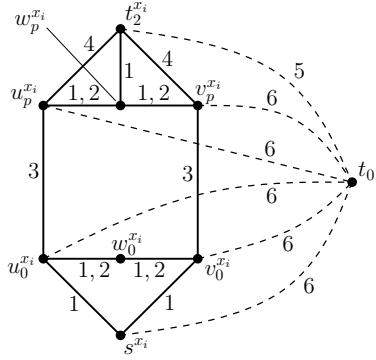} \hspace{0,2cm}
\caption{The addition of vertex $t_{0}$. There exists in $G_{\protect\phi }$ also the edge $  \{  t_{0}  ,  w_{0}^{x_{n}}  \}  $ with label~$\protect L _{\protect\phi }(\{  t_{0}  ,  w_{0}^{x_{n}}  \}  )=\{5\}$.}
\label{removal-t0-vertex-gadget-fig}
\end{subfigure}
\begin{subfigure}[b]{.8\linewidth}
\centering \includegraphics[scale=0.6]{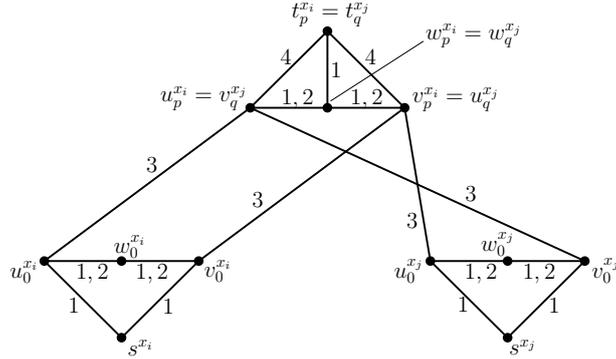}
\caption{The gadget for the clause $(x_{i}\oplus x_{j})$.}
\label{removal-clause-gadget-fig}
\end{subfigure}
\caption{Construction of $G_\phi$ and $L_\phi$.}
\label{removal-gadgets-fig}
\end{figure}

For every $i\in \{1,2,\ldots ,n\}$ the graph $G_{\phi ,i}$ has $16$
vertices. Furthermore, for every $p\in \{1,2,3\}$, the $4$ vertices of the $p
$th branch of $G_{\phi ,i}$ also belong to a branch of $G_{\phi ,j}$, for
some $j\neq i$. Therefore, together with the vertex $t_{0}$, the graph $%
G_{\phi }$ has in total $10n+1$ vertices. We now present the auxiliary
lemmas~\ref{total-number-labels-lem}-\ref{lambda-necessary-labels-lem} which
are necessary for the proof of Theorem~\ref%
{cost-removing-labels-upper-lower-bound-thm}.

\begin{lemma}
\label{total-number-labels-lem}
The labelling $L _{\phi }$ assigns $\frac{17}{4}n^{2}+28n+1$ labels to the edges of $G_{\phi}$.
\end{lemma}

\begin{proof}
The vertex $t_{0}$ has in total $3$ incident edges (to vertices $%
s^{x_{i}},u_{0}^{x_{i}},v_{0}^{x_{i}}$) to every base of a variable $x_{i}$
of $\phi $, $3$ incident edges (to vertices $%
t_{p}^{x_{i}},u_{p}^{x_{i}},v_{p}^{x_{i}}$, where $1\leq p\leq 3$) to every
clause $(x_{i}\oplus x_{j})$ of $\phi $ (i.e.,~to one branch to $x_{i}$ and
one branch of $x_{j}$), and one incident edge to vertex $w_{0}^{x_{n}}$.
That is, $t_{0}$ has in total $3n+3m+1=3n+3\cdot \frac{3}{2}n+1=\frac{15}{2}%
n+1$ incident edges, each of them having one label~in $L _{\phi }$.

Furthermore there exist in total $\frac{m(m-1)}{2} $ edges among the
vertices $\{t_{p}^{x_{i}}:1\leq i\leq n,\ 1\leq p\leq 3\}$, as well as $%
\frac{(n+m)(n+m-1)}{2}$ edges among the vertices $\{w_{p}^{x_{i}}:1\leq
i\leq n,\ 0\leq p\leq 3\}$, each of them having one label~in $L _{\phi
}$. Therefore, since $m=\frac{3}{2}n$, $L _{\phi }$ assigns in total $%
\frac{17}{4}n^{2}- 2n$ labels for these edges.

Moreover, the labelling $L _{\phi }$ assigns to every variable $x_{i}$
of $\phi $ in total $12$ labels, i.e.,~two labels for each of the transition
edges $   \{ u_{0}^{x_{i}} , w_{0}^{x_{i}} \}   , \    \{   w_{0}^{x_{i}}  ,  v_{0}^{x_{i}}   \}$ and one
label~for each of the edges $\{   \{  s^{x_{i}}  ,  u_{0}^{x_{i}}   \}    ,\    \{  s^{x_{i}}  ,   v_{0}^{x_{i}}  \} ,\   \{ u_{0}^{x_{i}}  ,  u_{p}^{x_{i}}  \} , \  \{  v_{0}^{x_{i}}  ,  v_{p}^{x_{i}}  \}  :1\leq p\leq 3\}$.

Finally, $L _{\phi }$ assigns to every clause $(x_{i}\oplus x_{j})$ of 
$\phi $ in total $7$ labels, i.e.,~two labels for each of the transition
edges $  \{ u_{p}^{x_{i}}  ,  w_{p}^{x_{i}}  \},\ \{  w_{p}^{x_{i}}  ,  v_{p}^{x_{i}}  \}  $ and one
label~for each of the edges $   \{u_{p}^{x_{i}}  ,  t_{p}^{x_{i}}  \} ,\
  \{ v_{p}^{x_{i}}  ,  t_{p}^{x_{i}}  \},\ \{  t_{p}^{x_{i}}  ,  w_{p}^{x_{i}}  \}   $, where $x_{i}$ is
associated with the $p$th branch of~$G_{\phi ,i}$. That is, $L _{\phi
} $ assigns in total $7m=\frac{21}{2}n$ labels for all clauses of $\phi $.

Summarizing, the labelling $L _{\phi }$ assigns to the edges of the graph $G_{\phi }$ a total of
$\left( \frac{15}{2}n+1\right) +\left( \frac{17}{4}n^{2}- 2n\right) +12n+%
\frac{21}{2}n=\frac{17}{4}n^{2}+28n+1$ labels.
\end{proof}

\begin{lemma}
\label{lambda-reachabilities-preserve-lem}
The labelling $L_{\phi }$ satisfies TC on $G_{\phi} $.
\end{lemma}

\begin{proof}
We will prove that there exists a temporal path in $L _{\phi }$
between any pair of vertices of $V_{\phi }=A\cup B\cup C\cup \{t_{0}\}$. 

For any two vertices $b,b^{\prime }\in B$ there exists a temporal path from $%
b$ to $b^{\prime }$ and from $b^{\prime }$ to $b$, due to the edge $%
\{ b , b^{\prime }  \}$ with label~$7$. Similarly, for any two vertices $c,c^{\prime
}\in C$ there exists a temporal path from $c$ to $c^{\prime }$ and from $%
c^{\prime }$ to $c$, due to the edge $  \{  c  ,  c^{\prime }  \}  $ with label~$7$. Let $%
a_{1},a_{2}\in A$. There exists a temporal path from $a_{1}$ to $a_{2}$ as
follows: start from $a_{1}$, follow $P_{i,p}$ (or $Q_{i,p}$) upwards until $%
t_{p}^{x_{i}}$ with greatest label~$4$, then go to $t_{0}$ with label~$5$,
and finally from $t_{0}$ to $a_{2}$ with label~$6$. In the special case
where $a_{1}$ and $a_{2}$ lie on the same path $P_{i,p}$ (resp. $Q_{i,p}$)
and $a_{1}$ appears before $a_{2}$ in $P_{i,p}$ (resp. $Q_{i,p}$), there
exists clearly a temporal path from $a_{1}$ to $a_{2}$ along $P_{i,p}$
(resp. $Q_{i,p}$).

Let $a\in A$ and $b\in B$. Note that $b=w_{p}^{x_{i}}$ for some $i\in
\{1,2,\ldots ,n\}$ and some $p\in \{0,1,2,3\}$. There exists the temporal
path from $b$ to $a$ as follows. First follow the edge $%
\{  w_{p}^{x_{i}}  ,  u_{p}^{x_{i}}  \} $ (with label~$1$), then follow upwards the path $%
P_{i,p}$ until one of the vertices $%
\{t_{1}^{x_{i}},t_{2}^{x_{i}},t_{3}^{x_{i}}\}$ (with maximum label~$4$),
then go to $t_{0}$ with label~$5$ and finally reach $a$ with label~$6$.
Furthermore there exists the temporal path from $a$ to $b$ as follows.
Assume first that $a=s^{x_{i}}$, for some $i\in \{1,2,\ldots ,n\}$. If $%
b=w_{0}^{x_{i}}$ then there exists the temporal path on the edges $%
\{  s^{x_{i}}  ,  u_{0}^{x_{i}}  \}  $ (with label~$1$) and $  \{  u_{0}^{x_{i}}  ,  w_{0}^{x_{i}}  \}$
(with label~$2$). If $b\neq w_{0}^{x_{i}}$ then there exists the temporal
path from $s^{x_{i}}$ to $w_{0}^{x_{i}}$ (with maximum label~$2$), followed
by the edge $\{  w_{0}^{x_{i}}  ,  b  \}  $ (with label~$7$). Assume now that $a\neq
s^{x_{i}}$, for every $i\in \{1,2,\ldots ,n\}$. That is, $a=u_{p}^{x_{i}}$
or $a=v_{p}^{x_{i}}$, for some $i\in \{1,2,\ldots ,n\}$ and some $p\in
\{0,1,2,3\}$. If $b=w_{p}^{x_{i}}$ then there exists the temporal path from $%
a$ to $b$ on the edge $\{ a , b \} $ (with label~$1$). If $b\neq w_{p}^{x_{i}}$ then
there exists the temporal path from $a$ to $b$ through the edges $%
\{ a , w_{p}^{x_{i}} \} $ (with label~$1$) and $  \{ w_{p}^{x_{i}}  ,  b  \} $ (with label~$7$). That
is, there exists a temporal path in $L _{\phi }$ between any $a\in A$
and any $b\in B$.

Let $b\in B$, i.e.,~$b=w_{p}^{x_{i}}$ for some $i\in \{1,2,\ldots ,n\}$ and
some $p\in \{0,1,2,3\}$. Then there exists a temporal path from $b$ to every
vertex $c\in C$ as follows. If $p=0$ then start with the edge $%
\{ w_{0}^{x_{i}} , u_{0}^{x_{i}}  \} $ (of label~$1$), continue upwards with a temporal
path (of maximum label~$4$) until $t_{1}^{x_{i}}\in C$ and then continue to
any other vertex $c\in C$ with the edge $  \{ t_{1}^{x_{i}} , c \} $ (of label~$7$). If $%
p\in \{1,2,3\}$ then reach $t_{p}^{x_{i}}\in C$ with the edge $%
\{ w_{p}^{x_{i}}  ,  t_{p}^{x_{i}}  \} $ (of label~$1$) and continue to any other vertex $%
c\in C$ with the edge $ \{  t_{p}^{x_{i}} , c  \}  $ (of label~$7$). That is, there exists
a temporal path from any $b\in B$ to any vertex of the set $C$. Now let $%
c\in C$, i.e.,~$c=t_{p}^{x_{i}}$ for some $i\in \{1,2,\ldots ,n\}$ and some $%
p\in \{1,2,3\}$. Then there exists a temporal path from $c$ to every vertex $%
b\in B$ as follows. First reach the vertex $w_{p}^{x_{i}}\in B$ with the
edge $ \{ t_{p}^{x_{i}}  ,  w_{p}^{x_{i}}  \} $ (of label~$1$) and then continue to any
other vertex $c\in C$ with the edge $  \{ w_{p}^{x_{i}}  ,  c  \} $ (of label~$7$). That
is, there exists a temporal path in $L _{\phi }$ between any $b\in B$
and any $c\in C$.

Let $a\in A$, i.e.,~$a\in \{s^{x_{i}},u_{p}^{x_{i}},v_{p}^{x_{i}}\}$ for some 
$i\in \{1,2,\ldots ,n\}$ and some $p\in \{0,1,2,3\}$. Then there exists at
least one path from $a$ upwards to a vertex $c\in
\{t_{1}^{x_{i}},t_{2}^{x_{i}},t_{3}^{x_{i}}\}$ (with maximum label~$4$).
Once we have (temporally) reached $c$ from $a$, we can (temporally) continue
to any other $c^{\prime }\in C$ through the edge $\{ c , c^{\prime } \} $ (of label~$7$%
). That is, there exists a temporal path from any $a\in A$ to any vertex of $%
C$. Now let $c\in C$, i.e.,~$c=t_{p}^{x_{i}}$ for some $i\in \{1,2,\ldots
,n\} $ and some $p\in \{1,2,3\}$. Then there exists a temporal path from $c$
to every vertex $a\in A$ as follows. First reach the vertex $t_{0}$ with the
edge $ \{  t_{p}^{x_{i}}  ,  t_{0}  \}$ (of label~$5$) and then continue to any vertex $%
a\in A$ with the edge $ \{  t_{0} , a \}$ (of label~$6$). That is, there exists a
temporal path in $L _{\phi }$ between any $a\in A$ and any $c\in C$.

Finally, there exists a temporal path between $t_{0}$ and every vertex of $%
A\cup C\cup \{w_{0}^{x_{n}}\}$, since $t_{0}$ is a neighbour with all these
vertices. Let $b\in B$, i.e.,~$b=w_{p}^{x_{i}}$ for some $i\in \{1,2,\ldots
,n\}$ and some $p\in \{0,1,2,3\}$. Then there exists a temporal path from $%
w_{p}^{x_{i}}$ to $t_{0}$ with the edges $ \{ w_{p}^{x_{i}}  ,  u_{p}^{x_{i}}   \}  $ (with
label~$1$) and $   \{  u_{p}^{x_{i}}  ,  t_{0}  \}  $ (with label~$6$). On the other hand,
there exists a temporal path from $t_{0}$ to every vertex $%
b=w_{p}^{x_{i}}\in B$, as follows. First reach the vertex $w_{0}^{x_{n}}$
with the edge $  \{  t_{0}  ,  w_{0}^{x_{n}}  \}  $ (of label~$5$) and then, if $b\neq
w_{0}^{x_{n}}$, continue with the edge $  \{  w_{0}^{x_{n}}  ,  b  \}  $ (of label~$7$). That
is, there exists a temporal path in $L _{\phi }$ between $t_{0}$ and
any vertex in $A\cup B\cup C$.

Summarizing, there exists a temporal path between any pair of vertices of $%
V_{\phi }=A\cup B\cup C\cup \{t_{0}\}$, i.e.,~the labelling $L
_{\phi }$ satisfies TC on $G_{\phi }$.
\end{proof}

\begin{lemma}
\label{lambda-necessary-labels-lem}
Let $L \subseteq L _{\phi }$
be a labelling of the graph $G_{\phi }$. If $L$ satisfies TC on $G_{\phi }$, then $L $ contains:
\begin{enumerate}[label=(\alph*)]
\item at least one label~for every transition edge $%
\{ u_{p}^{x_{i}} , w_{p}^{x_{i}}  \} $ and $ \{ w_{p}^{x_{i}} , v_{p}^{x_{i}} \}$, where $i\in
\{1,2,\ldots ,n\}$ and $p\in \{0,1,2,3\}$,

\item the label~of each edge $ \{ t_{p}^{x_{i}} , w_{p}^{x_{i}} \} $, where $i\in
\{1,2,\ldots ,n\}$ and $p\in \{1,2,3\}$,

\item the labels of all edges of $G_{\phi }$ among the vertices $%
\{t_{p}^{x_{i}}:1\leq i\leq n,\ 1\leq p\leq 3\}$,

\item the labels of all edges among the vertices $\{w_{p}^{x_{i}}:1\leq
i\leq n,\ 0\leq p\leq 3\}$,

\item the label~of each edge incident to $t_{0}$, and

\item the labels of all edges of the path $P_{i,p}$ or the labels of
all edges of the path $Q_{i,p}$, where $i\in \{1,2,\ldots ,n\}$ and $p\in
\{1,2,3\}$.
\end{enumerate}
\end{lemma}

\begin{proof}
\begin{enumerate}[label=(\alph*)]
\item First assume that $L $ does not keep any time label~on
the transition edge $ \{ u_{p}^{x_{i}}  ,  w_{p}^{x_{i}}  \}$ (resp. $%
\{ w_{p}^{x_{i}} , v_{p}^{x_{i}}  \}  $), where $i\in \{1,2,\ldots ,n\}$ and $p\in
\{0,1,2,3\}$. Then there does not exist in $L $ any temporal path from 
$u_{p}^{x_{i}}$ (resp. $v_{p}^{x_{i}}$) to $w_{p}^{x_{i}}$, even if $L 
$ maintains all other edge labels from $L _{\phi }$. This is a
contradiction. Therefore $L $ keeps at least one label~on the
transition edge $ \{  u_{p}^{x_{i}} , w_{p}^{x_{i}}  \}$ (resp. $%
\{ w_{p}^{x_{i}} , v_{p}^{x_{i}}  \}  $).

\item Now assume that $L $ does not contain the label~of some
edge $ \{ t_{p}^{x_{i}}  ,  w_{p}^{x_{i}}  \}  $, where $i\in \{1,2,\ldots ,n\}$ and $p\in
\{1,2,3\}$. Then there does not exist in $L $ any temporal path from $%
t_{p}^{x_{i}}$ to any vertex $w_{q}^{x_{j}}\in B$, even if $L $
maintains all other edge labels from $L _{\phi }$. This is a
contradiction to the assumption that $L $ satisfies TC on $G_{\phi }$. Therefore $L $ contains the label~of each edge $%
\{  t_{p}^{x_{i}} , w_{p}^{x_{i}}  \}  $, where $i\in \{1,2,\ldots ,n\}$ and $p\in
\{1,2,3\}$.

\item Consider two vertices $t_{p}^{x_{i}}\neq t_{q}^{x_{j}}$, $1\leq i<j\leq n$, $1\leq p,q\leq 3$. If $L $ does not contain the label
of the edge $ \{  t_{p}^{x_{i}}  ,  t_{q}^{x_{j}}  \}  $, then there does not exist in $%
L $ any temporal path from $t_{p}^{x_{i}}$ to $t_{q}^{x_{j}}$, which is a
contradiction. Therefore $L $ contains the labels of all edges of $%
G_{\phi }$ among the vertices $\{t_{p}^{x_{i}}:1\leq i\leq n,\ 1\leq p\leq
3\}$.

\item Assume that $L $ does not contain the label~of the edge $%
\{  w_{p}^{x_{i}} , w_{q}^{x_{j}}  \}  $, for some $i,j\in \{1,2,\ldots ,n\}$ and$\
p,q\in \{0,1,2,3\}$. Then there does not exist in $L $ any temporal
path from $w_{p}^{x_{i}}$ to $w_{q}^{x_{j}}$, which is a contradiction.
Therefore $L $ contains the labels of all edges among the vertices $%
\{w_{p}^{x_{i}}:1\leq i\leq n,\ 0\leq p\leq 3\}$.

\item We now prove that $L$ contains the label~of each edge
incident to $t_{0}$. Recall that the neighbours of $t_{0}$ in $G_{\phi }$ are
exactly the vertices of the set $A\cup C\cup \{w_{0}^{x_{n}}\}$. Assume $%
L $ does not have the label~of the edge $e= \{ t_{0}  ,   w_{0}^{x_{n}}   \}  $. Then
there exists no temporal path in $L $ from $t_{0}$ to any vertex $%
w_{p}^{x_{i}}\in B$, even if $L $ maintains all other edge labels from 
$L _{\phi }$. This is a contradiction to the assumption that $L $
satisfies TC on $G_{\phi }$. Now assume that there exists a
vertex $a\in A=\{s^{x_{i}},u_{p}^{x_{i}},v_{p}^{x_{i}}:1\leq i\leq n,\ 0\leq
p\leq 3\}$ such that $L $ does not have the label~of the edge $e=  \{  t_{0}  ,  a  \}
$. Then there does not exist in $L $ any temporal path from vertex $%
t_{0}$ to vertex $a$, which is again a contradiction. Finally assume that
there exists a vertex $t_{p}^{x_{i}}\in C$, such that $L $ does not
have the label~of the edge $e=  \{ t_{0}  ,  t_{p}^{x_{i}}  \}  $. Then there does not exist
in $L $ any temporal path from vertex $u_{p}^{x_{i}}$ to vertex $%
s^{x_{i}}$, which is a contradiction. Therefore $L $ contains the
label~of each edge incident to $t_{0}$.

\item Assume that $L $ misses from $L _{\phi }$ at least
one label~of the path $P_{i,p}$ and at least one label~of the path $Q_{i,p}$%
, for some $i\in \{1,2,\ldots ,n\}$ and $p\in \{1,2,3\}$. Then there does
not exist any temporal path from $s^{x_{i}}$ to $t_{p}^{x_{i}}$, which is a
contradiction. Therefore $L $ contains the labels of all edges of the
path $P_{i,p}$ or the labels of all edges of the path $Q_{i,p}$, where $i\in
\{1,2,\ldots ,n\}$ and $p\in \{1,2,3\}$.
\end{enumerate}
\end{proof}

We are now ready to provide the proof of Theorem~\ref{cost-removing-labels-upper-lower-bound-thm}.

\begin{theorem}
\label{cost-removing-labels-upper-lower-bound-thm}
There exists a truth assignment $\tau $ of $\phi $ with $|\tau (\phi )|\geq k$ if and only if there exists a TC satisfying labelling $L\subseteq L_{\phi}$ of $G_{\phi}$ such that $|L_{\phi} \setminus L|\geq 9n+k$.
\end{theorem}

\begin{proof}
($\Rightarrow $) Assume that there is a truth assignment $\tau $ that
XOR-satisfies $k$ clauses of $\phi $. We construct a labelling $L $ of $%
G_{\phi }$ by removing $9n+k$ labels from $L _{\phi }$, as follows.
First we keep in $L $ all labels of $L _{\phi }$ on the edges
incident to $t_{0}$. Furthermore we keep in $L $ the label~$\{7\}$ of
all the edges $  \{  t_{p}^{x_{i}} ,  t_{q}^{x_{j}}  \}  $ and the label~$\{7\}$ of all the
edges $w_{p}^{x_{i}}w_{q}^{x_{j}}$. Moreover we keep in $L $ the label~
$\{1\}$ of all the edges $  \{  t_{p}^{x_{i}}  ,  w_{p}^{x_{i}}  \}  $. Let now $i=1,2,\ldots
,n$. If $x_{i}=0$ in $\tau $, we keep in $L $ the labels of the edges
of the paths $P_{i,1},P_{i,2},P_{i,3}$, as well as the label~$1$ of the edge 
$  \{  v_{0}^{x_{i}}  ,  w_{0}^{x_{i}}  \}  $ and the label~$2$ of the edge $%
 \{ w_{0}^{x_{i}}  ,  u_{0}^{x_{i}}  \}  $. Otherwise, if $x_{i}=1$ in $\tau $, we keep in $%
L $ the labels of the edges of the paths $Q_{i,1},Q_{i,2},Q_{i,3}$, as
well as the label~$1$ of the edge $  \{  u_{0}^{x_{i}}  ,  w_{0}^{x_{i}}  \}  $ and the label~$2$ of the edge $  \{  w_{0}^{x_{i}}  ,  v_{0}^{x_{i}}  \}  $.

We now continue the labelling $L $ as follows. Consider an arbitrary
clause $\alpha =(x_{i}\oplus x_{j})$ of $\phi $. Assume that the variable $%
x_{i}$ (resp.~$x_{j}$) of the clause $\alpha $ corresponds to the $p$th
(resp.~to the $q$th) appearance of variable $x_{i}$ (resp.~$x_{j}$) in $\phi 
$. Then, by the construction of $G_{\phi }$, the $p$th branch of $G_{\phi
,i} $ coincides with the $q$th branch of $G_{\phi ,j}$, i.e.,~$%
u_{p}^{x_{i}}=v_{q}^{x_{j}}$, $v_{p}^{x_{i}}=u_{q}^{x_{j}}$, $%
w_{p}^{x_{i}}=w_{q}^{x_{j}}$, and $t_{p}^{x_{i}}=t_{q}^{x_{j}}$ (cf.~Figure~\ref{removal-clause-gadget-fig}). Let $\alpha $ be XOR-satisfied in $\tau $,
i.e.,~$x_{i}=\overline{x_{j}}$. If $x_{i}=\overline{x_{j}}=0$ (i.e.,~$x_{i}=0$
and $x_{j}=1$) then we keep in $L $ the label~$1$ of the edge $%
\{  v_{p}^{x_{i}}  ,  w_{p}^{x_{i}}  \}  $ and the label~$2$ of the edge $%
\{  w_{p}^{x_{i}}  ,  u_{p}^{x_{i}}  \}  $, cf.~Figure~\ref{removal-assignment-fig-1}. In
the symmetric case, where $x_{i}=\overline{x_{j}}=1$ (i.e.,~$x_{i}=1$ and $%
x_{j}=0$), we keep in $L $ the label~$1$ of the edge $%
\{  u_{p}^{x_{i}}  ,  w_{p}^{x_{i}}  \}  $ and the label~$2$ of the edge $%
\{  w_{p}^{x_{i}}  ,  v_{p}^{x_{i}}  \}  $. Let now $\alpha $ be XOR-unsatisfied in $\tau $%
, i.e.,~$x_{i}=x_{j}$. Then, in both cases where $x_{i}=x_{j}=0$ and $%
x_{i}=x_{j}=1$, we keep in $L $ the label~$1$ of both edges $%
\{  v_{p}^{x_{i}}  ,  w_{p}^{x_{i}}  \}  $ and $  \{  w_{p}^{x_{i}}  ,  u_{p}^{x_{i}}  \}  $, cf.~Figure~\ref%
{removal-assignment-fig-2}. This finalizes the labelling $L $ of $%
G_{\phi }$. It is easy to check that $L $ satisfies TC on $G_{\phi }$.

\begin{figure}[tbh]
\centering
\begin{subfigure}[t]{.99\linewidth}
\centering \includegraphics[scale=0.6]{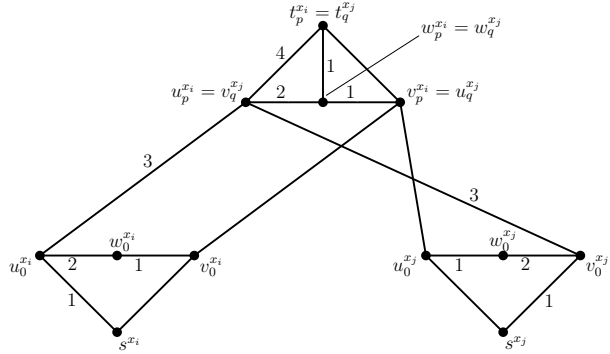} \hspace{0,2cm}
\caption{Case~$x_{i}=\overline{x_{j}}=0$.}
\label{removal-assignment-fig-1}
\end{subfigure}

\begin{subfigure}[b]{.99\linewidth}
\centering \includegraphics[scale=0.6]{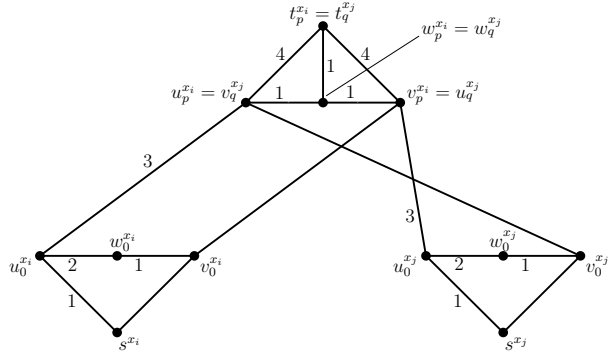}
\caption{Case~$x_{i}=x_{j}=0$.}
\label{removal-assignment-fig-2}
\end{subfigure}

\caption{The labelling $\protect L \subseteq \protect L_{\protect%
\phi}$ of the edges of fig.~\protect\ref{removal-clause-gadget-fig} for
the clause $\protect\alpha =(x_{i}\oplus x_{j})$ of $\protect\phi$.}
\label{removal-assignment-fig}
\end{figure}

Summarizing, the labelling $L $ misses in total $6$ labels of $L _{\phi }$ for the edges $%
\{     \{  s^{x_{i}}  ,  u_{0}^{x_{i}}  \}  ,\   \{  s^{x_{i}}  ,  v_{0}^{x_{i}}  \}  ,\
  \{  u_{0}^{x_{i}}  ,  w_{0}^{x_{i}}  \}  ,\   \{  w_{0}^{x_{i}}  ,  v_{0}^{x_{i}}  \}  ,\
  \{  u_{0}^{x_{i}}  ,  u_{p}^{x_{i}}  \}  ,\   \{  v_{0}^{x_{i}}  ,  v_{p}^{x_{i}}  \}  :1\leq p\leq 3,~ i = 1,2,\ldots ,n \}$.
That is, $L $ misses in total $6n$ labels of $L _{\phi }$ for
all variables $x_{1},x_{2},\ldots ,x_{n}$. For each of the $k$ XOR-satisfied
clauses $(x_{i}\oplus x_{j})$ of $\phi $, the labelling $L $ misses in
total $3$ labels of $L _{\phi }$ for the edges $%
 \{  u_{p}^{x_{i}}  ,  w_{p}^{x_{i}}  \}  ,\   \{  w_{p}^{x_{i}}  ,  v_{p}^{x_{i}}  \}  ,\
  \{  u_{p}^{x_{i}}  ,  t_{p}^{x_{i}}  \}  ,\   \{  v_{p}^{x_{i}}  ,  t_{p}^{x_{i}}  \}  ,\
  \{  t_{p}^{x_{i}}  ,  w_{p}^{x_{i}}  \} $, where $x_{i}$ is associated with the $p$th
branch of~$G_{\phi ,i}$. That is, $L $ misses in total $3k$ labels of $%
L _{\phi }$ for all XOR-satisfied clauses. Furthermore, for each of
the $m-k$ XOR-satisfied clauses $(x_{i}\oplus x_{j})$ of $\phi $, the
labelling $L $ misses in total $2$ labels of $L _{\phi }$ for the
edges $  \{  u_{p}^{x_{i}}  ,  w_{p}^{x_{i}}  \}  ,\   \{  w_{p}^{x_{i}}  ,  v_{p}^{x_{i}}  \}  ,\
  \{  u_{p}^{x_{i}}  ,  t_{p}^{x_{i}}  \}  ,\   \{  v_{p}^{x_{i}}  ,  t_{p}^{x_{i}}  \}  ,\
  \{  t_{p}^{x_{i}}  ,  w_{p}^{x_{i}}  \}  $, where $x_{i}$ is associated with the $p$th
branch of~$G_{\phi ,i}$. That is, $L $ misses in total $2(m-k)=3n-2k$
labels of $L _{\phi }$ for all XOR-satisfied clauses. All other labels
of $L _{\phi }$ remain in the labelling $L \subseteq L
_{\phi }$. Therefore, $L $ misses in total $6n+3k+3n-2k=9n+k$ labels
from $L _{\phi }$.

($\Leftarrow $) Assume that $r(G_\phi,L_\phi )\geq 9n+k$ and let $%
L \subseteq L _{\phi }$ be a TC preserving labelling of $G_{\phi }$ with $%
|L _{\phi }\setminus L |= r(G_\phi, L_\phi) \geq 9n+k$, i.e., $G_\phi (L)$ is minimal. Let ${i\in \{1,2,\ldots ,n\}}
$. For every $p\in \{1,2,3\}$, $L $ contains by Lemma~\ref%
{lambda-necessary-labels-lem}(f) the labels of all edges of the path $P_{i,p}
$ \emph{or} the labels of all edges of the path $Q_{i,p}$. Therefore, there
exist at least two indices $p_{1},p_{2}\in \{1,2,3\}$ such that $L $
contains the labels of all edges of the paths $P_{i,p_{1}},P_{i,p_{2}}$ 
\emph{or} the labels of all edges of the paths $Q_{i,p_{1}},Q_{i,p_{2}}$.
Without loss of generality let $p_{1}=1$ and $p_{2}=2$ and let $L $
contain the labels of all edges of the paths $P_{i,1},P_{i,2}$ (the other
cases can be dealt with in the same way by symmetry). Assume that $L $ also
contains the labels of all edges of the path $%
Q_{i,3}=(s^{x_{i}},v_{0}^{x_{i}},v_{3}^{x_{i}},t_{3}^{x_{i}})$. Then we can
modify the labelling $L $ to a labelling $L ^{\prime }$ as
follows. First remove from $L $ the labels of the edges $%
\{  s^{x_{i}}  ,  v_{0}^{x_{i}}  \}  $ and $  \{  v_{0}^{x_{i}}  ,  v_{3}^{x_{i}}  \}  $ and add instead the
labels of the edges $  \{  u_{0}^{x_{i}}  ,  u_{3}^{x_{i}}  \}  $ and $%
\{  u_{3}^{x_{i}}  ,  t_{3}^{x_{i}}  \}  $ (if they do not exist yet in $L $).
Furthermore change the labels of the transition edges $%
\{  v_{0}^{x_{i}}  ,  w_{0}^{x_{i}}  \}  $ and $  \{  w_{0}^{x_{i}}  ,  u_{0}^{x_{i}}  \}  $ to the labels $1
$ and $2$, respectively. Note that in the resulting labelling $L
^{\prime }$, both edges $  \{  u_{3}^{x_{i}}  ,  t_{3}^{x_{i}}  \}  $ and $%
\{  v_{3}^{x_{i}}  ,  t_{3}^{x_{i}}  \}  $ are labelled. Furthermore $L ^{\prime
}\subseteq L _{\phi }$ and $L ^{\prime }$ does not have more
labels than $L $, and thus $|L _{\phi }\setminus L
^{\prime }|\geq |L _{\phi }\setminus L |=r(G_\phi,L_\phi )$.
Moreover, it is easy to check that $L ^{\prime }$ still satisfies TC on $G_{\phi }$, as $L $ satisfies TC as
well. So, it must also be $|L_\phi \setminus L ^{\prime }| = r(G_\phi, L_\phi)$, i.e., $G_\phi(L ^{\prime })$ is also minimal. Therefore, we may assume without loss of generality that for any
minimal labelling $L \subseteq L _{\phi }$, $L $ contains
the labels of all edges of the paths $P_{i,1},P_{i,2},P_{i,3}$ \emph{or} the
labels of all edges of the paths $Q_{i,1},Q_{i,2},Q_{i,3}$. 

From Lemma~\ref{lambda-necessary-labels-lem}(a), $L$ contains at least $2n+2m$ labels on the edges of the form $\{u_p^{x_i}, w_p^{x_i}\}$ or $\{w_p^{x_i}, v_p^{x_i}\}$, since there are exactly $2n$ transition edges on the different bases of $G_\phi$ and $2m$ transition edges on the different branches of $G_\phi$. From Lemma~\ref{lambda-necessary-labels-lem}(b), $L$ contains $m$ additional labels, one for each branch, more specifically for the respective edge $\{t_p ^{x_i}, w_p ^{x_i} \}$ of the branch. From Lemma~\ref{lambda-necessary-labels-lem}(c), $L$ contains $\frac{m(m-1)}{2}$ extra labels among the vertices $\{t_p^{x_i}: 1 \leq i \leq n, 1 \leq p \leq 3\}$.
From Lemma~\ref{lambda-necessary-labels-lem}(d), $L$ also contains $\frac{(n+m) (n+m-1)}{2}$ additional labels among the vertices $\{w_p^{x_i}: 1\leq i \leq n, 0\leq p \leq 3 \}$. From Lemma~\ref{lambda-necessary-labels-lem}(e), $L$ also contains $\frac{15}{2}n+1$ labels on the edges incident to $t_0$. Finally, from Lemma~\ref{lambda-necessary-labels-lem}(f), $L$ contains at least $4n+m$ additional labels: for each $G_{\phi,i}$, $L$ contains at least $4$ labels, namely one label on the base edge $\{s^{x_i}, u_0^{x_i}\}$ or on the base edge $\{s^{x_i}, v_0^{x_i}\}$ and, for every $p\in \{1,2,3\}$, one label on the edge $\{u_0^{x_i},u_p^{x_i}\}$ or on the edge $\{v_0^{x_i},v_p^{x_i}\}$; also, for each branch of $G_\phi$, $L$ contains at least $1$ label, namely a label on the edge $\{u_p^{x_j}, t_p^{x_j}\}$ or on the edge $\{v_p^{x_j}, t_p^{x_j}\}$, for some $p\in \{1,2,3\}$ and $j \in \{1,2,\ldots, m\}$.

Notice that all the labels of $L$ mentioned above are on different edges, so no subset of labels has been accounted for more than once. Therefore, since $m=\frac{3n}{2}$, $L$ contains at least:
\begin{equation}
c(L) \geq \frac{17}{4} n^2 +17n +\frac{n}{2} +1
\label{eq-1}
\end{equation}
labels.

Now we construct from the labelling $L \subseteq L _{\phi }$ a
truth assignment $\tau $ for the formula $\phi $ as follows. For every $i\in
\{1,2,\ldots ,n\}$, if $L $ contains the labels of all edges of the
paths $P_{i,1},P_{i,2},P_{i,3}$, then we define $x_{i}=0$ in $\tau $.
Otherwise, if $L $ contains the labels of all edges of the paths $%
Q_{i,1},Q_{i,2},Q_{i,3}$, then we define $x_{i}=1$ in $\tau $. We will prove
that $|\tau (\phi )|\geq k$, i.e.,~that $\tau $ XOR-satisfies at least $k$
clauses of the formula $\phi $.

Let $\alpha =(x_{i}\oplus x_{j})$, where ${i,j\in \{1,2,\ldots ,n\}}$, be a clause of $\phi $ that is \emph{not} XOR-satisfied by $\tau$ in $\phi$.  Let $x_{i}$ (resp. $x_{j}$) be associated with the $%
p$th (resp. $q$th) branch of~$G_{\phi ,i}$ (resp. of~$G_{\phi ,j}$).
Since $\alpha$ is \emph{not} XOR-satisfied,
either $x_{i}=x_{j}=0$ or $x_{i}=x_{j}=1$ in $\tau $. If $x_{i}=x_{j}=0$ in $\tau $, it follows by the definition of
the assignment $\tau $ that the labelling $L $ contains the labels of
all edges of the path $P_{i,p}$ \emph{and} of the path $P_{j,q}$. Therefore, the $p^{th}$ branch of $G_{\phi,i}$, which is identified with the $q^{th}$ branch of $G_{\phi,j}$, has both edges $\{t_p ^{x_i}, u_p ^{x_i} \} \equiv \{t_q ^{x_j}, v_q ^{x_j} \}$ and $\{t_p ^{x_i}, v_p ^{x_i} \} \equiv \{t_q ^{x_j}, u_q ^{x_j} \}$ labelled under $L$, with one label each. The same holds if $x_{i}=x_{j}=1$, where all edges of both paths $Q_{i,p}$ \emph{and} $Q_{j,q}$ are labelled. So, for all the branches of $G_\phi$ that correspond to non-satisfied clauses of $\phi$ by the truth assignment $\tau$, $L$ contains an additional label (to the ones accounted for by using the result of Lemma~\ref{lambda-necessary-labels-lem}(f)). The number of clauses that are not satisfied by $\tau$ in $\phi$ is exactly $m-|\tau(\phi)|= \frac{3}{2} n - |\tau(\phi)|$.

Thus, it follows by Equation~(\ref{eq-1}), by adding the extra $\frac{3}{2} n - |\tau(\phi)|$, that $L$ contains in total at least:
\begin{eqnarray*}
c(L) &\geq& \frac{17}{4} n^2 +17n +\frac{n}{2} +1 + (\frac{3n}{2} - |\tau(\phi)|) \\
     &  = & \frac{17}{4} n^2 +19n +1 - |\tau(\phi)|
\end{eqnarray*}
labels.

Recall now that we have already shown in Lemma~\ref{total-number-labels-lem} that $L_\phi$ has a total of $\frac{17}{4} n^2 +28n +1$ labels. Therefore, we have:
\[|L _{\phi }\setminus L | = c(L_\phi) - c(L) \leq 9n+|\tau (\phi )|.\]

However, by our initial assumption:
\[|L _{\phi }\setminus L | =r(G_\phi, L_\phi) \geq 9n+k.\]

Therefore $9n+k\leq |L _{\phi }\setminus L |\leq
9n+|\tau (\phi )|$, and thus $|\tau (\phi )|\geq k$, i.e.,~the truth
assignment $\tau $ satisfies at least $k$ clauses of $\phi $. This completes
the proof of the theorem.
\end{proof}

The next corollary follows immediately by Theorem \ref%
{cost-removing-labels-upper-lower-bound-thm}.

\begin{corollary}
\label{cost-removing-labels-cor}Let OPT$_{\text{mon-Max-XOR}(3)}(\phi )$ the
greatest number of clauses that can be simultaneously XOR-satisfied by a
truth assignment of $\phi $. Then $r(G_{\phi },L_{\phi })=9n+$OPT$_{\text{%
mon-Max-XOR}(3)}(\phi )$.
\end{corollary}

\begin{proof}
Let $\tau $ be a truth assignment that satisfies $k=$OPT$_{\text{mon-Max-XOR}%
(3)}(\phi )$ clauses of $\phi $. Then there exists by Theorem \ref%
{cost-removing-labels-upper-lower-bound-thm} a TC satisfying labelling $L\subseteq L_{\phi }
$ of $G_{\phi }$ such that $|L_{\phi }\setminus L|\geq 9n+k$. Thus, since $%
r(G_{\phi },L_{\phi })\geq |L_{\phi }\setminus L|$, it follows that $%
r(G_{\phi },L_{\phi })\geq 9n+$OPT$_{\text{mon-Max-XOR}(3)}(\phi )$.
Conversely, let $L\subseteq L_{\phi }$ be a labelling of $G_{\phi }$ such
that $|L_{\phi }\setminus L|=r(G_{\phi },L_{\phi })$. Then there exists by
Theorem \ref{cost-removing-labels-upper-lower-bound-thm} a truth assignment $%
\tau $ that satisfies at least $r(G_{\phi },L_{\phi })-9n$ clauses of $\phi $%
. Thus OPT$_{\text{mon-Max-XOR}(3)}(\phi )\geq r(G_{\phi },L_{\phi })-9n$,
which completes the proof.
\end{proof}

Using Theorem~\ref{cost-removing-labels-upper-lower-bound-thm} and Corollary~\ref{cost-removing-labels-cor}, we are now
ready to prove the main theorem of this section.

\begin{theorem}
\label{cost-removing-labels-APX-hard-thm}The problem of computing $r(G,L)$ on an undirected temporally connected graph $G(L)$ is APX-hard.
\end{theorem}

\begin{proof}
Denote by OPT$_{\text{mon-Max-XOR}(3)}(\phi )$ the greatest number of
clauses that can be simultaneously XOR-satisfied by a truth assignment of $%
\phi $. The proof is done by an \emph{L-reduction} \cite{papadimitriou91}\
from the monotone Max-XOR(3) problem, i.e. by an approximation preserving
reduction which linearly preserves approximability features. For such a
reduction, it suffices to provide a polynomial-time computable function $g$
and two constants $\alpha ,\beta >0$ such that:

\begin{itemize}
\item $r(G_{\phi },L_{\phi })\leq \alpha \cdot $OPT$_{\text{mon-Max-XOR}(3)}(\phi )$, for any monotone XOR(3) formula~$\phi$, and

\item for any TC satisfying labelling $L\subseteq L_{\phi }$ of $G_{\phi }$, $g(L)$ is a
truth assignment for $\phi $ and OPT$_{\text{mon-Max-XOR}(3)}(\phi
)-|g(L)|\leq \beta \cdot (r(G_{\phi },L_{\phi })-|L_{\phi }\setminus L|)$,
where $|g(L)|$ is the number of clauses of $\phi $ that are satisfied by $%
g(L)$.
\end{itemize}

We will prove the first condition for $\alpha =13$. Note that a random truth
assignment XOR-satisfies each clause of $\phi $ with probability $\frac{1}{2}
$, and thus there exists an assignment $\tau $ that XOR-satisfies at least $%
\frac{m}{2}$ clauses of $\phi $. Therefore OPT$_{\text{mon-Max-XOR}(3)}(\phi
)\geq \frac{m}{2}=\frac{3}{4}n$, and thus $n\leq \frac{4}{3}$OPT$_{\text{%
mon-Max-XOR}(3)}(\phi )$. Now Corollary~\ref{cost-removing-labels-cor}
implies that:
\begin{eqnarray}
r(G_{\phi },L_{\phi }) &=&9n+\text{OPT}_{\text{mon-Max-XOR}(3)}(\phi ) 
\notag \\
&\leq &9\cdot \frac{4}{3}\text{OPT}_{\text{mon-Max-XOR}(3)}(\phi )+\text{OPT}%
_{\text{mon-Max-XOR}(3)}(\phi )  \label{apx-hardness-eq-1} \\
&=&13\cdot \text{OPT}_{\text{mon-Max-XOR}(3)}(\phi )  \notag
\end{eqnarray}

To prove the second condition for $\beta =1$, consider an arbitrary labelling 
$L\subseteq L_{\phi }$ of $G_{\phi }$. As described in the ($\Leftarrow $%
)-part of the proof of Theorem \ref%
{cost-removing-labels-upper-lower-bound-thm},\ we construct in polynomial
time a truth assignment $g(L)=\tau $ that satisfies at least $|L_{\phi
}\setminus L|-9n$ clauses of $\phi $, i.e. $|g(L)|=|\tau (\phi )|\geq
|L_{\phi }\setminus L|-9n$. Then:
\begin{eqnarray}
OPT_{\text{mon-Max-XOR}(3)}(\phi )-|g(L)| &\leq &OPT_{\text{mon-Max-XOR}%
(3)}(\phi )-|L_{\phi }\setminus L|+9n  \notag \\
&=&r(G_{\phi },L_{\phi })-9n-|L_{\phi }\setminus L|+9n
\label{apx-hardness-eq-2} \\
&=&r(G_{\phi },L_{\phi })-|L_{\phi }\setminus L|  \notag
\end{eqnarray}

This completes the proof of the Theorem. 
\end{proof}

\begin{note*}
In fact, we have also shown (Theorem~\ref{cost-removing-labels-upper-lower-bound-thm}) that the problem of computing the removal profit is NP-hard in the strong sense, since all numbers used in the reduction are constant integers.
\end{note*}

\subsection{Temporally connected random labellings have high removal profit}\label{sec:random_labels_minimal}

In this section, we show that dense graphs with random labels have the property TC and have a very high removal profit asymptotically almost surely.
More specifically, we consider the complete graph and the Erd\"os-Renyi model of random graphs, $G_{n,p}$ and we examine whether we can delete labels from such temporal graphs and continue preserving TC.

The (single-labelled) model of temporal graphs that we consider here is that of \emph{uniform random temporal graphs}~\cite{akrida-jpdc}.
\begin{definition}\label{def:random}\hspace{-0,01cm}\protect\cite{akrida-jpdc}
A \emph{uniform random temporal graph} is a graph $G$ on $n$ vertices, $n \in \mathbb{N}$, each edge of which receives exactly one label uniformly at random from a set $\{1,2,\ldots, \alpha\},~\alpha \in \mathbb{N}^*$ and the selection of the label of an edge is independent from the selection of the label of any other edge.
\end{definition}

\subsubsection{High removal profit in the complete graph}

\begin{theorem}\label{thm:random_complete}
In the uniform random temporal graph where the underlying graph $G$ is the complete graph (clique) of $n$ vertices and $\alpha \geq 4$, we can delete all but $\Theta(n\log{n} + \log^2{n})$ labels without violating TC, with probability at least $1-\frac{1}{n^2}$.
\end{theorem}
\begin{proof}
First, note that any set $\{1,2, \ldots, \alpha\}$ of $\alpha$ consecutive natural numbers can be partitioned into $4$ disjoint almost equal subsets of consecutive numbers, $A_1,A_2,A_3,A_4$. Indeed, let $\alpha= 4k+v$, where $k\in \mathbb{N}$ and $v \in \{1,2,3,4\}$.

For $v=0$, we use $A_1=\{1,\ldots,k\}, A_2=\{k+1, \ldots, 2k\}, A_3=\{2k+1, \ldots, 3k\}, A_4= \{3k+1, \ldots, 4k\}$.

For $v=1$, we use $A_1=\{1,\ldots,k\}, A_2=\{k+1, \ldots, 2k\}, A_3=\{2k+1, \ldots, 3k\}, A_4= \{3k+1, \ldots, 4k+1\}$.

For $v=2$, we use $A_1=\{1,\ldots,k\}, A_2=\{k+1, \ldots, 2k\}, A_3=\{2k+1, \ldots, 3k+1\}, A_4= \{3k+2, \ldots, 4k+2\}$.

For $v=3$, we use $A_1=\{1,\ldots,k\}, A_2=\{k+1, \ldots, 2k+1\}, A_3=\{2k+2, \ldots, 3k+2\}, A_4= \{3k+3, \ldots, 4k+3\}$.

In any of the above four cases, each particular edge of the clique $K_n$ receives a single random label $l$, with:
\[Pr[l \in A_i] \geq \frac{k}{4k+3},~ \forall i=1,2,3,4\]
Since $k\geq 1$ (because $\alpha\geq 4$), we have $\frac{k}{4k+3} \geq \frac{1}{7}$. So, we get the following Lemma:
\begin{lemma}\label{lem:clique_set_partition}
For each particular edge $e$ of $K_n$ and for the label $l$ that it receives, it holds that $Prob[l \in A_i] \geq \frac{k}{4k+3},~ \forall i=1,2,3,4$.
\end{lemma}
Now, colour \emph{green}($g$), \emph{yellow}($y$), \emph{blue}($b$) and \emph{red}($r$) the edges that are assigned a label in $A_1$, $A_2$, $A_3$ and $A_4$ respectively.

\begin{definition}
A temporal router (cf.~Figure~\ref{fig:router_complete}) of a clique $G=K_n=(V,E)$ is a subgraph $R=(V_R,E_R)$ of $G$, with $2\gamma \log{n}+1$ vertices, $\gamma$ being a constant such that $\gamma \geq 4 \cdot \frac{1}{\log_2{\frac{2500}{2499}}}$, with the following properties (all logarithms are with base $2$ here):
\begin{enumerate}[label=\alph*)]
\item $V_R$ is the union of a particular vertex $v_0$ (called the \emph{centre} of $R$) and two equisized vertex sets $V_{in}$ and $V_{out}$, each of $\gamma \log{n}$ vertices, and
\item $R$ is the induced subgraph of $G$ formed from $V_R$ (so it is a clique itself).
\end{enumerate}
\begin{figure}[!htb]
\centering
\includegraphics[width=0.45\textwidth]{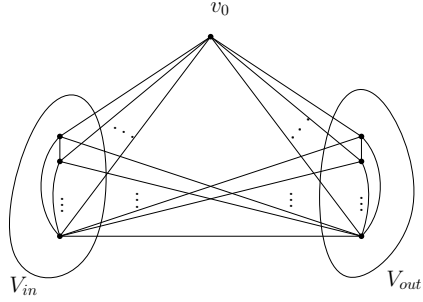}
\caption{Temporal router of a clique}
\label{fig:router_complete}
\end{figure}
Note that $R$ has $|E_R|=2 \gamma \log{n} + \frac{(2 \gamma \log{n})\cdot(2 \gamma \log{n}-1)}{2}$ edges.
\end{definition}

Let $w, w'$ be any two vertices of the clique that are not in $V_R$. We consider the edges connecting $w$ to $V_{in}$ and the edges connecting $w'$ to $V_{out}$; using those edges and the edges of $R$, there are $\gamma \log{n}$ edge-disjoint paths of length $4$ (each) connecting $w$ and $w'$. Let us call those paths \emph{special paths} and note that every such path uses edges of the form $\{w,v_{in}\}, \{v_{in}, v_0\}, \{v_0, v_{out}\}, \{v_{out}, w'\}$, where $v_{in} \in V_{in}$ and $v_{out} \in V_{out}$ (cf.~Figure~\ref{fig:special_path}).
\begin{figure}[!htb]
\centering
\includegraphics[width=0.7\textwidth]{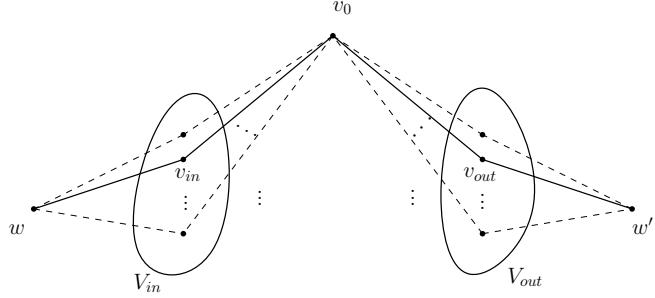}
\caption{A special path connecting $w$ and $w'$.}
\label{fig:special_path}
\end{figure}

Each special path $P=(w,v_{in},v_0,v_{out},w')$ connecting $w$ and $w'$ becomes a $(w,w')$-journey if the label $l_1$ of $\{w,v_{in}\}$ is in $A_1$, the label $l_2$ of $\{v_{in}, v_0\}$ is in $A_2$, the label $l_3$ of $\{v_0, v_{out}\}$ is in $A_3$, and the label $l_4$ of $\{v_{out}, w'\}$ is in $A_4$. Then, the probability that $P$ is a journey is at least $\left( \frac{1}{7} \right)^4$, due to independence of the labels' selection.

Since all special paths that connect $w$ and $w'$ are edge-disjoint, the probability that none of them is a $(w,w')$-journey is:
\begin{eqnarray*}
Pr[\text{no special path is a }(w,w') \text{-journey}] = \left( 1- \frac{1}{7^4} \right)^{\gamma \log{n}} &<& \left( 1- \frac{1}{2500} \right)^{\gamma \log{n}} \\
 &=&  n^{-\gamma \log{\frac{2500}{2499}}}.
\end{eqnarray*}

Therefore, we have:
\begin{lemma}\label{lem:no_journey}
For any two particular vertices $w,w'$ of $V\setminus V_R$, the probability that there is a special path $P$ from $w$ to $w'$ that is a $(w,w')$-journey is at least $1-n^{-\gamma \log_2{\frac{2500}{2499}}}$.
\end{lemma}

Now, we consider only the edges and labels of $R$ and, for each $w\in V\setminus V_R$, we consider only the edges connecting $w$ to each vertex of $R$; the sparsified graph $G'=(V,E')$ has, thus, $|E'|= \frac{(2\gamma \log{n} +1) \cdot 2 \gamma \log{n}}{2} +\big( n-(2\gamma \log{n}+1) \big) \cdot (2\gamma \log{n} +1)=\Theta(n\log{n}+ \log^2{n})$ edges. We will show that we need only consider the edges (and labels) of $G'$ to maintain $TC$ in $G$, i.e., that $G'$ itself is temporally connected, with probability at least $1-\frac{1}{n^2}$.

Consider any pair, $w,w'$, of vertices of the uniform random temporal graph on $K_n$ and a temporal router $R$. Also, consider the graph $G'$ as described above, with the labelling implied by the uniform random labelling on the clique. If $w,w'\in V_R$, then they are directly connected via a labelled edge in $G'$ and thus a journey exists both ways between them. If $w \in V_R$ and $w \in V\setminus V_R$, then again there is a direct labelled edge in $G'$ connecting $w$ and $w'$, so there is a journey between them either way.

It remains to examine the existence in $G'$ of journeys between pairs of vertices $w, w',~w\not= w'$, none of which is in $V_R$; there are at most $n^2$ such pairs of vertices. Under the random labelling on $G$, let $\mathcal{E}_1$ be the event that there exists a pair $w, w' \in V\setminus V_R$ such that there is no $(w,w')$-journey via a special path through $R$. Also, let $\mathcal{E}_2$ be the event that for a specific pair $w, w' \in V\setminus V_R$, there is no $(w,w')$-journey via a special path through $R$. Then,
\[ Pr[\mathcal{E}_1] \leq n^2 Pr[\mathcal{E}_2] \text{ (by the Union Bound).}\]

So, we have:
\[ Pr[G' \text{ is not temporally connected}] \leq n^2 n^{-\gamma \log_2{\frac{2500}{2499}}}. \]

Note that Lemma~\ref{lem:no_journey} gives an upper bound on the probability of the event $\mathcal{E}_2$. Set $\gamma$ to be $\gamma \geq 4 \cdot \frac{1}{\log_2{\frac{2500}{2499}}}$. Then, we have:

\[ Pr[G' \text{ is not temporally connected}] \leq n^{-2} . \]
\end{proof}

\subsubsection{High removal profit in dense random Erd\"os-Renyi graphs}

In this section, we consider the underlying graph $G=(V,E)$ to be an instance of the Erd\"os-Renyi graph model, $G_{n,p}$, with $p\geq \left( \frac{\gamma \ln{n}}{n} \right)^{\frac{1}{7}},~\gamma \geq 3$.

\begin{definition}[\textbf{Erd\"{o}s-Renyi graphs}]
An instance of $G_{n,p}$ is formed when for every pair of vertices $u,v$ among a total number of $n$ vertices, the edge $\{u,v\}$ is chosen to exist with probability $p$ independently of any other edge.
\end{definition}

Notice that $G_{n,p}$ is almost surely connected for any $p \geq 2 \frac{\ln{n}}{n}$~\cite{bollobasb}. As in the previous section, we consider here a uniform random temporal graph on $G$, i.e., we consider each edge of $G$ to receive exactly one label uniformly at random from a set $\{1,2,\ldots, \alpha\}$, with $\alpha \geq 4$. The selection of the label of an edge is independent of the selection of the label of any other edge. Also, the label selection process is independent of the process of selection of edges in $G_{n,p}$. As in Theorem~\ref{thm:random_complete}, we consider partitioning $\{1,2,\ldots, \alpha\}$ into \emph{four consecutive subsets, $A_1,A_2,A_3,A_4$, of consecutive positive integers}, where each subset is of size either $\lfloor \frac{\alpha}{4} \rfloor$ or $\lfloor \frac{\alpha}{4} \rfloor +1$; such a partition is always possible. Now colour \emph{green}($g$), \emph{yellow}($y$), \emph{blue}($b$) and \emph{red}($r$) the edges that are assigned a label in $A_1$, $A_2$, $A_3$ and $A_4$, respectively. As in Lemma~\ref{lem:clique_set_partition}, we have:
\begin{lemma}
For each particular edge of $G$ and for the label $l$ that it receives, it holds that $Prob[l \in A_i] \geq \frac{1}{7},~ \forall i=1,2,3,4$.
\end{lemma}

In such instances of $G_{n,p}$, we cannot assume the existence of cliques such as the clique of the temporal router used in the previous section. Indeed, even for very dense instances of $G_{n,p}$, with $p=\frac{1}{2}$, the largest clique is at most of size $2\ln{n}$~\cite{bollobasb}.

In order to ``sparsify'' labelled instances $G$ of $G_{n,p}$, by removing labels without violating TC, we need to guarantee the existence of much sparser routing subsets of $G$.

\begin{definition}
Given two vertices $v_1,v_2$ of $G_{n,p}$, a \emph{temporal router, $R(v_1,v_2)$, in an instance $I$ of $G_{n,p}$} is a subgraph of $I$ that has vertices $v_1,v_2$ and additional vertices $a_1,\ldots, a_k$ and $b_1, \ldots, b_k$ so that:
\begin{itemize}
\item $v_1$ connects directly to each $a_i,b_i,~i=1,\ldots, k$,
\item $v_2$ connects directly to each $a_i,b_i,~i=1,\ldots, k$,
\item each pair $a_i,b_i$ is directly connected, $i=1,\ldots, k$,
\item each edge $\{a_i,b_i\}$ receives a green label, $i=1,\ldots, k$,
\item each edge $\{a_i,v_1\}$ receives a yellow label, $i=1,\ldots, k$,
\item each edge $\{v_1,b_i\}$ receives a blue label, $i=1,\ldots, k$,
\item each edge $\{v_2,a_i\}$ receives a blue label, $i=1,\ldots, k$,
\item each edge $\{b_i,v_2\}$ receives a yellow label, $i=1,\ldots, k$.
\end{itemize}
\end{definition}

Figures~\ref{fig:random_router_k_1} and~\ref{fig:random_router_k_2} show a temporal router $R(v_1,v_2)$ for $k=1$ and $k=2$ respectively.
\begin{figure}[!htb]
\centering
\includegraphics[width=0.3\textwidth]{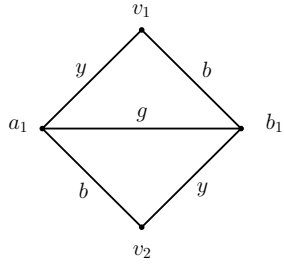}
\caption{Temporal router $R(v_1,v_2)$ for $k=1$.}
\label{fig:random_router_k_1}
\end{figure}

\begin{figure}[!htb]
\centering
\includegraphics[width=0.3\textwidth]{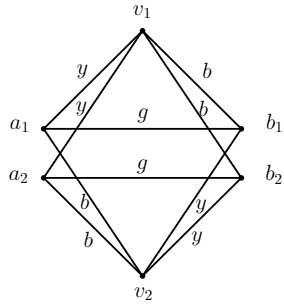}
\caption{Temporal router $R(v_1,v_2)$ for $k=2$.}
\label{fig:random_router_k_2}
\end{figure}

\begin{note*}
A temporal router $R(v_1,v_2)$ in an instance $I$ of $G_{n,p}$ is \emph{temporally connected} since:
\begin{itemize}
\item any $a_i$ can reach any $b_j$, via a journey through $v_1$, i.e., $(a_i,v_1,b_j)$ is a journey,
\item any $b_i$ can reach any $a_j$, via a journey through $v_2$, i.e., $(b_i,v_2,a_j)$ is a journey,
\item any $a_i$ can reach any $a_j\not= a_i$, via a journey through $b_i$ and then $v_2$, i.e., $(a_i,b_i,v_2,a_j)$ is a journey,
\item any $b_i$ can reach any $b_j\not= b_i$, via a journey through $a_i$ and then $v_1$, i.e., $(b_i,a_i,v_1,b_j)$ is a journey,
\item $v_1$ can reach $v_2$, via any $a_i$, i.e., $(v_1,a_i,v_2)$ is a journey,
\item $v_2$ can reach $v_1$, via any $b_i$, i.e., $(v_2,b_i,v_1)$ is a journey, and
\item all other (temporal) connections are direct.
\end{itemize}
\end{note*}

\begin{definition}
We denote by $R_i$ and call it the \emph{$i^{th}$ theta subgraph of $R(v_1,v_2)$} the labelled subgraph of $R(v_1,v_2)$ induced by the vertices $v_1,v_2,a_i$, and $b_i$, for some $i=1, \ldots, k$ (cf.~Figure~\ref{fig:theta_sub}).
\begin{figure}[!htb]
\centering
\includegraphics[width=0.3\textwidth]{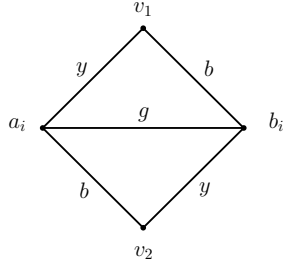}
\caption{The $i^{th}$ theta subgraph, $R_i$.}
\label{fig:theta_sub}
\end{figure}
\end{definition}

Note that the following Lemma holds:
\begin{lemma}\label{lem:rand_exp}
Let $I$ be an instance of $G_{n,p}$ with a uniform random labelling from the set $\{1,2,\ldots, \alpha\},~\alpha \geq 4$. Then, for each particular $i=1,\ldots, k$, it holds:
\[Pr[R_i \text{ exists in } I] \geq \left( \frac{p}{7} \right)^5\]
\end{lemma}
\begin{proof}
Each edge of $R_i$ is realized in $G_{n,p}$ with probability $p$ and receives the correct type of label (green, yellow, blue, or red) with probability at least $\frac{1}{7}$. Note also that the edges of different theta subgraphs $R_i$ and $R_j$, $i\not= j$, are disjoint. Thus, in $G_{n,p}$, the random experiments of each of the theta subgraphs $R_i$ appearing are \emph{independent} from each other and each succeeds with probability at least $\left( \frac{p}{7} \right)^5$.
\end{proof}

Now, consider the set of vertices $V \setminus \{v_1,v_2\}$ and partition it into two almost equal sets $V_1$ and $V_2$; note that $|V_i| \geq \lfloor \frac{n}{2} \rfloor -2 = n',~i=1,2$. Consider a pairing of $n'$ vertices of $V_1$ to $n'$ vertices of $V_2$ and let the $n'$ different pairs be the possible pairs of vertices $a_i,b_i$ in a theta subgraph $R_i$. By Lemma~\ref{lem:rand_exp} and since the random experiments are independent, the number of appearances of $R_i$ is at least the number of successes in a Bernoulli distribution of $n'$ trials, with success probability $\left( \frac{p}{7} \right)^5$ per trial. Therefore, by the Chernoff bound, we have the following Lemma:
\begin{lemma}\label{lem:prob_theta_sub}
In any instance $I$ of a $G_{n,p}$ that has been labelled uniformly at random, the probability that there is a temporal router $R(v_1,v_2)$ consisting of at least $k= \frac{n'}{2} \left( \frac{p}{7} \right)^5 $ theta subgraphs $R_i$ is at least $1- e^{-\frac{1}{4} n' \left( \frac{p}{7} \right)^5} \geq 1- e^{-\frac{n}{12} \left( \frac{p}{7} \right)^5} $ (since $n' \geq \frac{n}{3}$).
\end{lemma}

\begin{corollary}
Note that, again by the Chernoff bound, $k$ asymptotically almost surely \emph{does not exceed} $\frac{3}{2}n' \left( \frac{p}{7} \right)^5 $, since $Pr[k> \left(1+\frac{1}{2} \right) E(k)] \leq e^{-\frac{1}{4} \frac{n}{12} \left( \frac{p}{7} \right)^5}$.
\end{corollary}

We now condition on the event, $\mathcal{E}_1$, that the instance $I$ of a labelled $G_{n,p}$ has a temporal router $R(v_1,v_2)$ of at least $k=\frac{n}{6} \left( \frac{p}{7} \right)^5$ theta subgraphs. By Lemma~\ref{lem:prob_theta_sub}, we know that:
\[Pr[ 	\bar{\mathcal{E}_1}] \leq e^{-\frac{n}{12} \left( \frac{p}{7} \right)^5}.\]
Notice that to calculate the probability of the event $\mathcal{E}_1$, we did not examine at all the appearance of edges from a vertex outside of $R(v_1,v_2)$ to any vertex of $R(v_1,v_2)$. Given that $R(v_1,v_2)$ exists, any vertex $u$ that is not in $R(v_1,v_2)$ can reach any vertex $u'$ that is also not in $R(v_1,v_2)$ \emph{through $R(v_1,v_2)$} via a journey, if $u$ connects to some $a_i$ directly with a green edge, and  $b_i$ connects to $u'$ directly with a red edge. Then, $(u,a_i,v_1,b_i,u')$ is a journey.

The probability of the edge $\{u,a_i\}$ being green \emph{and} the edge $\{b_i,u'\}$ being red, for any $i$, is $\left( \frac{p}{7} \right)^2$ and it is independent of edge experiments inside $R(v_1,v_2)$. So, we have:

\begin{lemma}
Condition on the event $\mathcal{E}_1$ of the existence of $R(v_1,v_2)$ in $G_{n,p}$ with at least $k= \frac{n}{6} \left( \frac{p}{7} \right)^5$ theta subgraphs. Let $u,u'$ any two (different) vertices of $G$ that are not in $R(v_1,v_2)$. Then:
\[ Pr[\text{there exists a } (u,u') \text{-journey through } R(v_1,v_2)] \geq 1- \big(1- (\frac{p}{7})^2 \big)^k \]
\end{lemma}
\begin{proof}
For the vertices $u,u'$ as described above and any one of the $k$ possible journeys of the form $(u,a_i,v_1,b_i,u')$, the probability that such a journey fails (i.e., is not realized) is at most $1-\left( \frac{p}{7} \right)^2$. Therefore, given $\mathcal{E}_1$, we have:
\[ Pr[\text{there exists no } (u,u') \text{-journey through } R(v_1,v_2)] \leq \big(1- (\frac{p}{7})^2 \big)^k \]
\end{proof}

Let $\mathcal{E}_2$ be the event that given $k$ pairs of vertices $a_i,b_i$ in a possible $R(v_1,v_2)$, each vertex pair $u,u'$, with $u\not= u'$ and $u,u' \not \in V\big( R(v_1,v_2) \big)$, satisfies the following: there is at least one pair of vertices $a_i,b_i$ such that $u$ connects to $a_i$ with a green edge \emph{and} $u'$ connects to $b_i$ with a red edge.


Notice that $\bar{\mathcal{E}_2}$ is the event that there is a pair of vertices $u,u'$ that are not in $R(v_1,v_2)$ that fails to connect as described above. Since the number of possible pairs of vertices $u,u'$ is less than $n^2$, we have:
\begin{eqnarray*}
Pr[\bar{\mathcal{E}_1}] & \leq & n^2  \big(1- (\frac{p}{7})^2 \big)^k \\
 & \leq & n^2 e^{-k (\frac{p}{7})^2} \\
 & = & e^{-k (\frac{p}{7})^2 +2 \ln{n}}
\end{eqnarray*}

Note now that the following Lemma holds:
\begin{lemma}
The event $\mathcal{E}_1$ and the evnt $\mathcal{E}_2$ (given $\mathcal{E}_1$) guarantee that the (labelled) instance $I$ of $G_{n,p}$ is temporally connected ``via the temporal router $R(v_1,v_2)$''.
\end{lemma}
\begin{proof}
Condition on $\mathcal{E}_1$ and on $\mathcal{E}_2$ (given $\mathcal{E}_1$). Then, for each vertex $u\not \in V\big( R(v_1,v_2) \big)$, keep one of its green edges (to some $a_i$) and one of its red edges (to the corresponding $b_i$), since by $\mathcal{E}_2$, those exist. Remove all edges of $I$ except for the edges of $R(v_1,v_2)$ and the two edges we keep for every vertex that is not in $R(v_1,v_2)$. The resulting labelled subgraph of $I$ is temporally connected, since:
\begin{enumerate}[label=\alph*)]
\item $R(v_1,v_2)$ is temporally connected itself, by construction,
\item any $u \not \in V\big( R(v_1,v_2) \big)$ has a journey via $R(v_1,v_2)$ to any other $u' \in V$ in the graph,
\item any $a_i$ or $b_j$ can reach any $u \not \in V\big( R(v_1,v_2) \big)$ via a journey through $v_1$(using first a green edge, if we start from a $b_j$ vertex, and then using a yellow, a blue and a red edge to reach $u$), and
\item $v_1$ and $v_2$ can reach any $u \not \in V\big( R(v_1,v_2) \big)$ via a journey through some vertex $b_i$ (using first a blue -or yellow, respectively- edge to $b_i$, and then a red edge to $u$).
\end{enumerate}
\end{proof}

The temporally connected instance $I$ of $G_{n,p}$ after the removal of redundant edges as described above has a number of labelled edges (i.e., time-edges) that is at most $2n+\Theta(k)$. Since $k=\frac{n}{6} \left( \frac{p}{7}\right)^5$, $I$ has at most $\Theta(n+ n p^5)$ labels after the removal of the redundant edges.

We will now choose $p$ so that we have:
\begin{enumerate}[label=\alph*)]
\item $k \left( \frac{p}{7}\right)^2 \geq \gamma \ln{n},~\gamma >3$, and
\item $k= \Theta\big( n\left(\frac{p}{7}\right)^5 \big)$ with probability $1- e^{-\frac{n}{12} \left( \frac{p}{7} \right)^6}$.
\end{enumerate}

It suffices to have $\Theta \big(  n\left(\frac{p}{7}\right)^7 \big) \geq \gamma \ln{n}~(\gamma>3)$,i.e., \[p \geq 7 \left( \frac{\gamma \ln{n}}{n} \right)^{\frac{1}{7}} . \]

We conclude with the following Theorem:
\begin{theorem}
Consider a $G_{n,p}$, with $p \geq 7 \left( \frac{\gamma \ln{n}}{n} \right)^{\frac{1}{7}}$, for some $\gamma>3$, labelled uniformly at random. Then, any instance $I$ of $G_{n,p}$ needs only $\Theta( n + n p^5 )$ time-edges to be temporally connected, with probability at least $1-2e^{-\gamma' \ln{n}}$,for some $\gamma' \geq 1$.
\end{theorem}
\begin{proof}
Via the temporal router $R(v_1,v_2)$, any instance $I$ of $G_{n,p}$ becomes temporally connected by using at most $\Theta(n+n p^5)$ edges (and, thus, labels), with probability at least:
\begin{eqnarray*}
Pr[\mathcal{E}_1] \cdot Pr[\mathcal{E}_2 | \mathcal{E}_1] & \geq & \left( 1- e^{-\frac{n}{12} \left( \frac{p}{7} \right)^5} \right) \cdot \left( 1- e^{-k (\frac{p}{7})^2 +2 \ln{n}} \right)  \text{ (for }     k \geq \frac{n}{2} \left( \frac{p}{7} \right)^5  \text{)} \\
 & \geq &  1- 2e^{-\frac{n}{12} \left( \frac{p}{7} \right)^5}\\
 & \geq &  1- 2e^{-\gamma' \ln{n}} \text{ (by the chosen range of } p \text{).}
\end{eqnarray*}
\end{proof}

Note that for the sparsest possible $G_{n,p}$ here, i.e., for $p= 7 \left( \frac{\gamma \ln{n}}{n} \right)^{\frac{1}{7}} $, we need only 
$\Theta(n+n^{\frac{2}{7}} (\ln{n})^{\frac{5}{7}}) = \Theta(n)$
 edges (and, thus, labels) to satisfy TC, with probability at least $1- 2e^{-\gamma' \ln{n}},~ \gamma' \geq 1$.

\section{Conclusions and further research}
In this work, we study the complexity of testing and designing issues of nearly cost-optimal temporal networks that are temporally connected. It remains an open problem to provide a polynomial-time constant factor approximation algorithm for the computation of the removal profit in a given temporally connected temporal graph. Further research could also investigate the complexity of computing the removal profit in special classes of graphs, e.g., planar graphs or the grid. Extensions of this research also include the study of the interval temporal networks model, where edges can be available for continuous intervals of time, as well as a more in-depth study of models of random temporal networks.

\begin{acknowledgements*}

This work was supported in part by:
\begin{enumerate}[label=(i)]
\item the School of EEE/CS and its NeST initiative at the University of Liverpool
\item the FET EU IP Project MULTIPLEX under contract No. 317532, and
\item the EPSRC Grant EP/K022660/1.
\end{enumerate}
\end{acknowledgements*}

\bibliography{CoRR_design}  

\end{document}